\documentclass[journal]{IEEEtran}
\usepackage [english]{babel}
\usepackage [autostyle, english = american]{csquotes}
\MakeOuterQuote{"}
\pdfoutput=1

\usepackage{cite}
\usepackage[breaklinks=true]{hyperref}

\ifCLASSINFOpdf
\else
\fi
\usepackage[cmex10]{amsmath}
\usepackage{amssymb}
\usepackage{amsthm}
\newtheorem{theorem}{Theorem}

\usepackage{mathtools}

\usepackage{lipsum}

\newtheoremstyle{break}
  {\topsep}{\topsep}%
  {\itshape}{}%
  {\bfseries}{}%
  {\newline}{}%
\theoremstyle{break}

\newtheoremstyle{break}
  {}
  {}
  {\itshape}
  {}
  {\bfseries}
  {.}
  {\newline}
  {}

\theoremstyle{break}

\usepackage{tikz}

\begin{document}
%
\title{Dynamic Watermarking: Active Defense of Networked Cyber-Physical Systems}
%
%
%

\author{Bharadwaj Satchidanandan
        and~P.~R.~Kumar,~\IEEEmembership{Fellow~IEEE}
\thanks{This paper is partially based on work supported by NSF Science and Technology Center Grant CCF-0939370, NSF Contract No. CNS-1302182, the US Army Research Office under Contract No. W911NF-15-1-0279, and the AFOSR under Contract No. FA-9550-13-1-0008.}}
\maketitle

\begin{abstract}
The coming decades may see the large scale deployment of networked cyber-physical systems to address global needs in areas such as energy, water, healthcare, and transportation. However, as recent events have shown, such systems are vulnerable to cyber attacks. Being safety critical, their disruption or misbehavior can cause economic losses or injuries and loss of life. It is therefore important to secure such networked cyber-physical systems against attacks. In the absence of credible security guarantees, there will be resistance to the proliferation of cyber-physical systems, which are much needed to meet global needs in critical infrastructures and services.

This paper addresses the problem of secure control of networked cyber-physical systems. This problem is different from the problem of securing the communication network, since cyber-physical systems at their very essence need sensors and actuators that interface with the physical plant, and malicious agents may tamper with sensors or actuators, as recent attacks have shown. 

We consider physical plants that are being controlled by multiple actuators and sensors communicating over a network, where some sensors could be "malicious," meaning that they may not report the measurements that they observe. We address a general technique by which the actuators can detect the actions of malicious sensors in the system, and disable closed-loop control based on their information. This technique, called "watermarking," employs the technique of actuators injecting private excitation into the system which will reveal malicious tampering with signals. We show how such an active defense can be used to secure networked systems of sensors and actuators.
\end{abstract}


%
\IEEEpeerreviewmaketitle

\section{Introduction}
%
%
%
%
\IEEEPARstart{T}he $21^{st}$ century could well be the era of large-scale system building. Such large-scale systems are envisioned to be formed by the interconnection of many embedded devices communicating with each other, and interacting with the physical world. Their operation requires tight integration of communication, control, and computation, and they have been termed broadly as Cyber-Physical Systems (CPS). The smart energy grid, intelligent transportation systems, internet of things, telesurgical systems, and robotics are examples of such cyber-physical systems.

While the importance and benefits of cyber-physical systems require no emphasis, their sustained proliferation is contingent on some key challenges being addressed, security being a primary one. Since CPSs have many applications in safety-critical scenarios, security breaches of these systems can have adverse consequences including economic loss, injury and death. 

There have been many instances of demonstrated attacks on cyber-physical systems in the recent past \cite{cardenas_challenges,sinopoli}. In Maroochy-Shire, Australia, in the year 2003, a disgruntled ex-employee of a sewage treatment corporation hacked into the computers controlling the sewage system and issued commands which led to a series of faults in the system \cite{maroochyShire, cardenas_challenges}. This is an insider attack, where the adversary has the necessary credentials to access and issue control commands to the system. We will return to this point shortly. Another example is the attack on computers controlling the Davis-Besse nuclear power plant in Ohio. In the year 2003, the Slammer worm, which infected about $75000$ hosts in the internet in under ten minutes, also infected the computers controlling the nuclear power plant, disabling the safety monitoring systems \cite{cardenas_challenges}. While the Slammer worm was not designed to target the nuclear power plant, the use of commodity IT software in control systems made them vulnerable to such attacks \cite{cardenas_challenges}. Another pertinent example is the Stuxnet worm which, in the year 2010, exploited a vulnerability in Microsoft Windows to subvert critical computers controlling centrifuges in Iran's uranium enrichment facility \cite{stuxnet}. Having subverted the computers, it issued control commands that caused the centrifuges to operate at abnormally high speeds, causing them to tear themselves apart. In order to keep the attacks undetected by software-implemented alarm routines and officials in the control room, Stuxnet recorded the sensor values in the facility for twenty-one seconds before carrying out each attack, and replayed those twenty-one seconds in a constant loop during the attack. Stuxnet has been claimed to be the first known digital weapon \cite{stuxnet}, and since then, cyberwarfare has emerged as a serious concern for cyber-physical systems due to the many advantages it offers to the attacker such as allowing it to remain anonymous, attack without geographical constraints, etc. Today, the resources required to carry out such attacks on critical infrastructures are generally available \cite{lbnl}, underlining the urgent need for the research community to pay attention to this problem. 

In this paper, we examine the problem of detecting attacks on networked cyber-physical systems. These systems can be thought of as having two layers- a physical layer, which consists of the plant, actuators, controllers, and sensors, which interact with physical signals, and a cyber layer, which networks the components of the physical layer. While securing the cyber layer is certainly of importance, by itself, it does not constitute the security of the cyber-physical system as a whole. The Maroochy-Shire incident is a classic illustration of this point, where the malfunctioning of the plant was not the result of an attack on the network layer, but of authorized individuals attacking the physical layer by the issue of improper control commands. 

At first look, it appears as though securing the physical layer is harder than securing the cyber layer. In the problem of network security, there is a clear distinction between an honest party and an adversary. Attributes such as credentials and cryptographic keys distinguish honest parties from adversaries. However, when it comes to securing the physical layer, no such demarcation exists. Any authorized party is also a potential adversary. 

However, we show in this paper that what works in favor of securing the physical layer, and what we exploit, is the fact that the actions of every node interfacing with the physical layer get transformed into physical signals, and these signals can be subject to scrutiny for semantic consistency. To elaborate, consider a physical system consisting of a plant, some actuators, and some sensors. If an actuator injects into the system a probing signal that is not disclosed to other nodes in the system, then, combined with the knowledge of the plant dynamics, the actuator expects the signals to appear in transformed ways at various points in the system. Based on the information that the actuator receives from the sensors about the signals at various points, it can potentially infer if there is malicious activity in the system or not. 

We develop these ideas to secure the physical layer of a noisy dynamical system. We examine a protocol whereby honest actuator nodes deliberately superimpose certain stochastically independent probing signals on top of the control law they are intended to apply. We then propose specific "tests" that these actuator nodes perform to infer malicious activity, and establish their effectiveness. As a particular illustration of the results, for example, we can show that even if all the sensors are malicious, and the actuators have absolutely no measurements that they can directly make and have to completely rely on the malicious sensors for \emph{all} their purported measurements, then even in such an adverse environment, the actuators can under appropriate conditions ensure that the additional distortion on performance that the malicious sensors can cause is of mean-square zero, if they are to remain undetected. Using this approach of active defense, we establish that under appropriate conditions, no matter in what way the adversarial sensors collude, the amount of distortion that they can add without exposing their presence can have an average power of only zero. 

The method we examine is a dynamic version of "watermarking," \cite{watermarking1} where certain indelible patterns are imprinted into a medium that can detect tampering \cite{watermarking1,watermarking2,watermarking3}. It shows how one can watermark dynamic signals so that one can detect malicious misbehavior on the part of sensors or actuators. On top of a secure communication system, it provides overall security to a cyber-physical system against malicious sensors and actuators.   

This paper is organized as follows. Section \ref{literature} describes prior work in this area. Section \ref{problem} provides a system-theoretic formulation of the problem. Section \ref{solution} describes our approach of active defense for networked cyber-physical systems. Section \ref{siso} opens by describing the method in the relatively simple context of a scalar linear Gaussian system and rigorously establishes the associated theoretical guarantees. Section \ref{sisoarx} treats the more general class of scalar auto-regressive systems with exogenous noise (ARX systems) that is Gaussian. Section \ref{ARMAX} extends these ideas to the more general ARMAX systems with arbitrary delay, a model that is frequently encountered in process control. Section \ref{partially_observed} deals with partially observed SISO systems with Gaussian process and measurement noise. Section \ref{mimo} considers multi-input, multi-output linear Gaussian systems in state-space form. Section \ref{nongaussian} describes how our results can potentially be extended to non-Gaussian systems. Section \ref{statTests} shows how the theoretical results lead to statistical tests that can be used to detect malicious behavior within a delay bound with a controlled false alarm rate. Section \ref{conclusion} provides some concluding remarks.

\section{Prior Work}\label{literature}
The vulnerability and the need to secure critical infrastructure from cyberattacks has been recognized at least as early as in 1997 \cite{PresidentsReport}. Subsequent reports \cite{Early1,Early2,Early3,Early4,Early5,Early6} cited demonstrated attacks, identifying potential threats, and analyzing the effects of successful attacks on specific systems. The large-scale replacement of proprietary control software and protocols by commodity IT software and protocols by the industry in order to allow for interoperability and rapid scalability has increased the vulnerability of Industrial Control Systems (ICS) to cyberattacks, and roadmaps were prepared to address security of control systems in various sectors such as the energy sector \cite{energysector}, water sector \cite{watersector}, chemical sector \cite{chemicalsector}, and the transportation sector \cite{transportationsector}. 

Some of the initial work on secure control \cite{cardenas1} has addressed the definition of what constitutes a secure control system. Certain key operational goals such as closed-loop stability and performance metric of interest are noted in \cite{cardenas1}, and it is proposed that a secure control system must achieve these operational goals even when under attack, or at least cause only a gradual degradation. It also identifies how the problem of secure control of networked systems departs from the traditional problems of network security and information security. In the former, authorized users or insiders can launch attacks on the system causing physical damage, as in the Maroochy-Shire incident. Hence, network and information security measures such as intrusion prevention and detection, authentication, access control, etc., fundamentally cannot address these attacks. Therefore, securing the network does not amount to securing the NCS. In this paper, we build a framework on top of a secure communication network to secure the NCS. There has been recent work showing how one indeed can build a communication network that provides provable guarantees on security, throughput as well as delays \cite{ponniah1}.

A theoretical study of secure control benefits from having a model for the adversary, and \cite{cardenas2} defines certain adversary and attack models. A popular attack on communication networks is the Denial-of-Service (DoS) attack, in which the adversary floods the communication network with useless packets, rendering it incapable of transporting useful information. Another attack is the deception attack, where an adversary impersonates as another node and transmits false information on its behalf. In the context of a networked control system, the adversary could employ DoS attack to prevent the controller and actuator from receiving the data required for their operation, and could employ deception to cause an actuator node to issue incorrect actuation signals. A framework to study the evolution of the physical process under Denial-Of-Service (DoS) and deception attacks is presented in \cite{cardenas2}. 

Classical estimation algorithms such as the Kalman filter assume perfect communication between the sensor and the controller. However, in the presence of an unreliable network, which may stem from either the characteristics of the network or from adversarial presence, certain packets may be dropped. This affects the performance of the estimation and control algorithm. Motivated by this scenario, a large body of research has been devoted to developing estimation and control algorithms for systems with intermittent observations, c.f. \cite{PacketLoss1,PacketLoss2,PacketLoss3,PacketLoss4,PacketLoss5} and references therein. Though this line of work does not explicitly model adversarial behavior, some of these ideas have been employed in the literature of secure control. For instance, \cite{amin2009safe} studies, using some of the machinery developed in \cite{PacketLoss1} in the context of control over lossy communication networks, the effect of DoS attack on the control performance. Also addressed in the literature is the effect of deception attacks on estimation and control performance. In \cite{Mo_sinopoli_deception}, the effect of false data fed by the compromised sensors to the state estimator is studied. The goal is to characterize the set of all estimation biases that an adversary can inject without being identified. A related work is \cite{vgupta}, where fundamental trade-offs between the detection probability (for a fixed false alarm rate) and the estimation error that the adversary can cause, are studied for single-input-single-output systems.

Several techniques have also been developed to counter attacks on cyber-physical systems. A technique for correct recovery of the state estimate in the presence of malicious sensors is presented in \cite{diggavi1}, and also a characterization of the number of malicious sensors that can be tolerated by the algorithm. A well-known attack on control systems is the replay attack, where the adversary records the sensor measurements for a fixed period of time and replays them during the attack so as to maintain the illusion of a normal operating condition. It was shown in \cite{sinopoli_replay} that only systems for which the matrix $(A+BL)(I-KC)$ is stable are susceptible to replay attacks, where $L$ is the feedback gain and $K$ is the steady-state Kalman gain of the system's Kalman filter. Consequently, for such systems, a method to secure the system from replay attack is presented in \cite{sinopoli_replay,sinopoli_replay2,physical_watermarking1,physical_watermarking2}. The fundamental idea of these is to inject into the actuation signal a component that is not known in advance. Specifically, \cite{sinopoli_replay,physical_watermarking1} consider the replay attack, employed in Stuxnet, and introduce a technique, termed Physical Watermarking, wherein the controller commands the actuators to inject into the system a component that is random and not known in advance in order to secure the system against such an attack. To the best of our knowledge, this is the first use of the idea of watermarking. It is shown that by employing Physical Watermarking, the covariance of the innovations process when the system is "healthy" and that when it is under attack are significantly different, enabling the estimator to detect the attack using a $\chi_2$ detector. This technique is extended in \cite{physical_watermarking2} to detect an adversary employing more intelligent attack strategies. Specifically, the adversary is assumed to possess a set of capabilities, based on which a specific attack strategy, consisting of the adversary generating false measurement values that are reported to the estimator, is identified. It is shown that Physical Watermarking can counter such an adversary. Including a random component in the actuation signal would clearly affect the running cost, and \cite{pappasReplay} develops an optimal policy to switch between cost-centric and security-centric controllers, by formulating the problem as a stochastic game between the system and the adversary. A method to detect false-data injection is presented in \cite{revealingStealthyAttacks}, where the focus is on zero-dynamics attacks, attacks which cannot be detected based on input and output measurements. A method to verify the measurements received from multiple sensors is presented in \cite{shield}, which exploits correlations between sensor measurements and other features to weed out the measurements from malicious sensors. Though not presented in the context of a dynamical system, the ideas presented in \cite{shield} could in principle be extended to incorporate system dynamics. 

At a high level, the techniques generally proposed in the literature for secure control, such as estimation with intermittent or incorrect observations, exploiting known correlations between sensor measurements to weed out suspicious measurements, or techniques inspired from fault-tolerant control, can be classified as passive techniques, meaning that they anticipate that the adversary will inject malicious signals into the system, and therefore employ a design that minimizes the damage that can be so caused. On the contrary, in this paper, we pursue active defense along the lines of \cite{sinopoli_replay, physical_watermarking1,physical_watermarking2}, an alternative approach for secure control in which, over and above the control-policy specified excitation, the system is excited in ways unknown to the adversary, thereby preventing the adversary from injecting malicious signals into the system. This constitutes a form of watermarking of a dynamic system \cite{watermarking1, watermarking2,watermarking3}. Since it does not distinguish between an adversarial node and a faulty node, this approach also falls under the purview of fault-tolerant control. Such a singular departure from passive approach is \cite{sinopoli_replay}, \cite{physical_watermarking1}, \cite{physical_watermarking2}, where it has been proposed to inject a signal to guard against adversarial attacks. This paper develops a comprehensive treatment of such dynamic watermarking. To the best of our knowledge, the general theory and techniques of active defense for \emph{arbitrary} attacks does not appear to have been studied in the literature thus far. That is what is developed in this paper.

\section{Problem Formulation}\label{problem}
Fig. \ref{fig_ncs} illustrates the basic architecture of a Networked Cyber-Physical System. At the heart of the system is a physical plant with $m$ inputs and $n$ outputs. Each input is controlled by an independent actuator, and each output is measured by an independent sensor. A known transformation capturing the dynamics of the physical plant maps the actuation signals applied by the actuators to outputs that are measured by the sensors. These measurements are communicated to entities called controllers through an underlying communication network. Each controller's job is to compute, in accordance with a control policy, the particular actuation signals that must be applied. The result of this computation is then communicated to the actuators through the communication network, which then apply the actuation signals. We assume that the network is complete, so that every node in the network can communicate with every other node. We use the term "nodes" generically to refer to any entity in the network. Therefore, in Fig. \ref{fig_ncs}, some nodes are actuators, some are sensors, some are controllers, and some could just be relays whose only job is to forward the information from one node to another. In this work, we assume that each node has both communication and computational capabilities, thereby allowing the controllers to be collocated with the actuators. 

However, certain sensors in the system could be "malicious" (suggestive of this, some nodes in Fig. \ref{fig_ncs} are marked in red), and the other nodes are said to be "honest." We further assume that the malicious nodes know the identity of all other malicious nodes in the system, allowing them to collude to achieve their objective, whereas the honest nodes don't know which of the other nodes are malicious or honest. A malicious sensor is a sensor node which does not report accurately the measurements that it observes. Rather, it reports a distorted version of the measurements. A malicious router node may not forward the packets that it receives, may forward packets that it does not receive (while claiming otherwise), alter packets before forwarding, introduce intentional delays, impersonate some other node in the system, etc. Therefore, the challenges of securing a Networked Cyber-Physical System are two-fold:
\begin{itemize}
\item[1)] Secure the cyber layer that comprises the communication network, ensuring confidentiality, integrity, and availability of network packets, and
\item[2)] Secure the sensors and actuators interfacing with the physical layer.
\end{itemize}
The former is achieved by a combination of traditional approaches such as cryptography and a more recent line of work reported in \cite{ponniah1, ponniah4, ihong}, while the latter is the subject of this paper.
\begin{figure}
\centering
\includegraphics[width=\columnwidth]{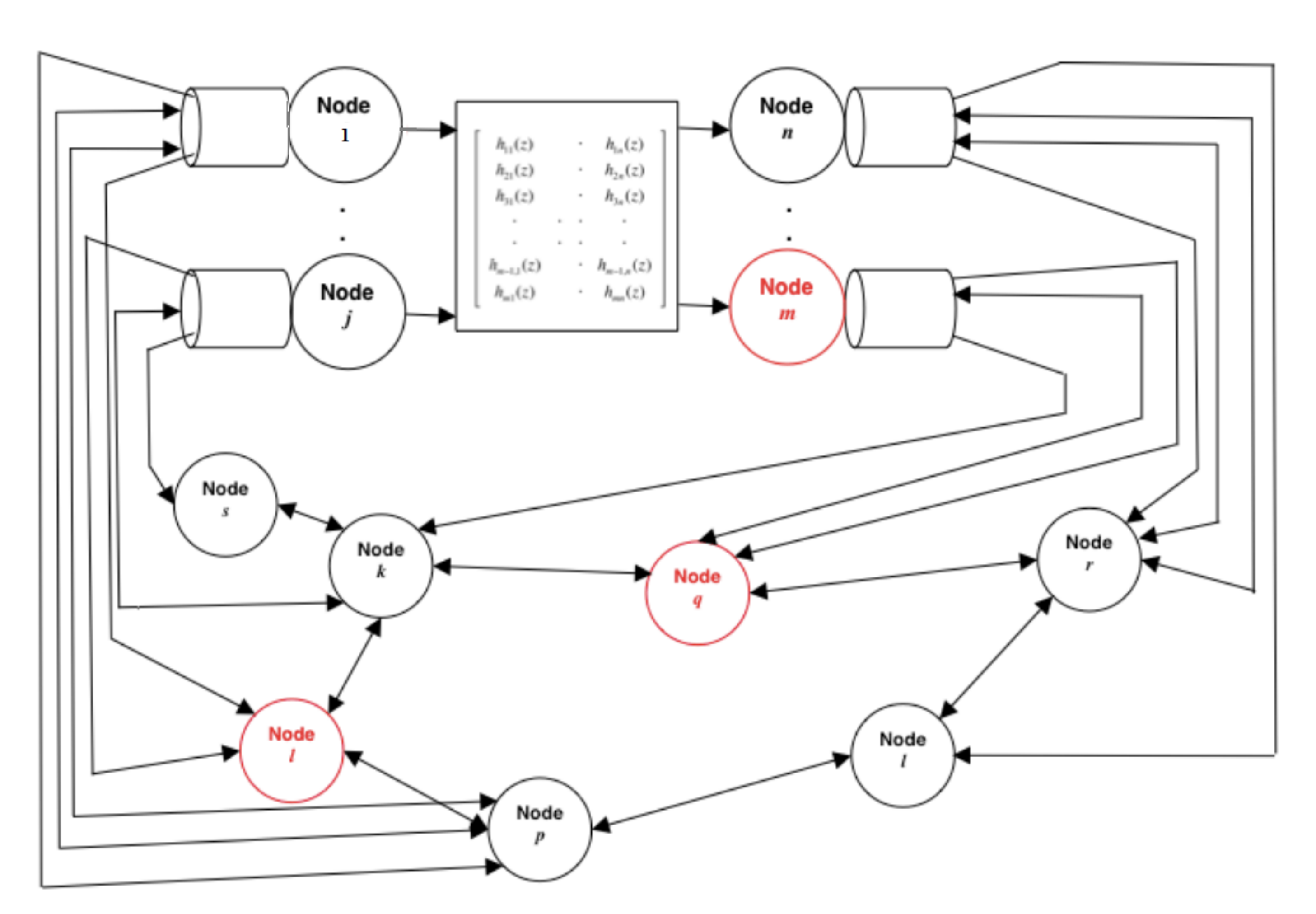}
\caption{A Networked Cyber-Physical System}\label{fig_ncs}
\end{figure}

Based on our assumption that the communication network is complete, and that the cyber layer has been secured, we suppose that every node in the network knows the identity of the node from which a packet that it receives originated, and knows if a packet that it receives was tampered with by any node along the route. Therefore, going forward, we abstract the cyber layer as consisting of secure, reliable, delay guaranteed bit pipes between any pair of nodes in the network. In particular, we imagine that there exists a secure, reliable, delay guaranteed bit pipe between any particular sensor and actuator node.

Finally, the plant that is being controlled is abstracted as a stochastic linear dynamical system described by time-invariant parameters. While any system that is of practical interest is most certainly non-linear, we focus on linear systems for two reasons. The first is that a linear system lends itself to tractable analysis, and enables one to separate the complexity arising out of the problem at hand from the complexity arising as a consequence of the system's non-linearity. Secondly, a theory developed for linear systems provides valuable insights and design principles that often transcend the particulars of the model and apply to a much broader class of systems. The wide applicability of Kalman's pioneering work on linear control systems \cite{Kalman} stands testimony to this fact. 

We are now in a position to state the problem in precise terms (with notation as indicated in the Appendix). Consider an $m\times n$ stochastic linear dynamical system of order $p$, described by 
\begin{align}
\mathbf{x}[t+1]=A\mathbf{x}[t]+B\mathbf{u}[t]+\mathbf{w}[t+1],\label{firsteq}
\end{align}
where $A\in \mathbb{R}^{p\times p}$, $B\in \mathbb{R}^{p\times m}$, $\mathbf{u}[t]$ is the input applied to the plant at time $t$, $\mathbf{\{w\}}$ is a sequence of independent and identically distributed (i.i.d.) Gaussian random vectors with zero mean and covariance martix $\sigma_w^2I$, independent of the initial state $\mathbf{x}_0$ of the system. 

While a malicious sensor can report different measurements to different actuators, the consistency of the reported measurements can be checked by allowing the actuators to exchange the reported measurements among themselves. This constrains the malicious sensors to report the same value to all honest nodes in the system. The sensor node $j$ reports to the honest nodes in the system the value $z_j[t]$ as the measurement that it observes at time $t$. We define $\mathbf{z}[t]\coloneqq [z_1[t] \;  z_2[t] \;  z_3[t] \cdots  z_n[t]]^T$. We will call $\mathbf{z}[t]$ for $t\geq 0$ the measurements \emph{reported} by the sensors. Note that the sensor $j$ is honest if $z_j[t]=x_j[t]$ $\forall t$. 

We assume that a control policy is in place, known to all nodes in the system, and allow for it to be history dependent, so that the $i^{th}$ input at time $t$, ${u_i^g}[t]$, dictated by the policy is 
\begin{equation}
{u_i^g}[t]=g_t^i(\mathbf{z}^{t}),\label{controlLaw}
\end{equation}
where $\mathbf{z}^{t}\coloneqq\{\mathbf{z}[0],\mathbf{z}[1],...,\mathbf{z}[t]\}$. We can suppose without loss of generality that the controller that computes this control law is collocated with the actuator node.

Our goal in this paper is to secure the control system by developing techniques that prevent the malicious nodes from causing excessive distortion if they are to remain undetected. We will suppose that the purpose of control law (\ref{controlLaw}) is to improve the regulation performance of the system with respect to the disturbances affecting it. For some of the performance results we will assume that the system is open-loop stable, and show that the adversarial nodes cannot affect the performance of the system without remaining undetected, and if detected then they can be disconnected, returning the system to stable behavior. The fundamental results characterizing what is the most that the adversarial nodes can do while remaining undetected, are applicable to all systems, stable or unstable. 

\section{Dynamic Watermarking: An Active Defense for Networked Cyber-Physical Systems}\label{solution}
The key idea which allows the honest nodes to detect the presence of malicious nodes in the system is the following. Let $g$ denote the control policy in place that specifies inputs to be applied in response to observed outputs. At each time instant $t$, an actuator node superimposes on its control policy-specified input $u_i^g[t]$, a random variable $e_i[t]$ that it draws independently from a specified distribution. Therefore, the input that actuator $i$ applies at time $t$ is 
\begin{equation}
u_i[t]=u_i^g[t]+e_i[t].
\end{equation}
This is illustrated in Fig. \ref{actuator_probing}. The random variables $e_0[t], e_1[t], e_2[t] \cdots$ are independent and identically distributed (i.i.d.), independent of the control policy-specified input. The distribution that they are chosen from is made public (i.e., made known to every node in the system), but the actual values of the excitation are not disclosed. In fact, this is how the honest actuators can check whether signals are being tampered with as they travel around the control loop. We refer to these random variables as an actuator node's \emph{privately imposed excitation}, since only that actuator node knows the actual realization of the sequence.

\begin{figure}
\begin{tikzpicture}

\draw[black, thick] (4,4) rectangle (7,8);
\draw[->, black, thick] (3,7.5) -- (4,7.5);
\draw[->, black, thick] (3,4.5) -- (4,4.5);
\draw (2,7.5) circle (1cm);
\draw (2.5,4.5) circle (0.5cm);
\draw[loosely dotted] (2.5,6.5) -- (2.5,5); 
\draw[black, thick] (2.5,7.5) -- (3,7.5);
\draw (2.25,7.5) circle (0.25cm);
\draw[->, black, thick] (1.5,7.5) -- (2,7.5);
\draw[->, black, thick] (2.25,8) -- (2.25,7.75);

\draw (8,7.5) circle (0.5cm);
\draw (8,4.5) circle (0.5cm);
\draw[->, black, thick] (7,7.5) -- (7.5,7.5);
\draw[->, black, thick] (7,4.5) -- (7.5,4.5);
\draw[loosely dotted] (8,7) -- (8,5);

\draw[black, thick] (2.25,7.65) -- (2.25,7.35);
\draw[black, thick] (2.1,7.5) -- (2.4,7.5);

\node[draw,align=center] at (3,10) {Private excitation $e_i[t]$ is\\ superimposed on the\\ control policy-specified input};
\draw[->, black, dotted] (3,9.2) -- (2.1,8.4);

\node[draw,align=center] at (7.3,10) {Actual actuation signal\\ contains private excitation\\ known only to actuator $i$};
\draw[->, black, dotted] (6,9.3) -- (3.1,7.6);

\node[align=center] at (5.5,6) {\textbf{PHYSICAL}\\ \textbf{PLANT}};

\node[align=center] at (8,7.5) {Node\\ s};
\node[align=center] at (8,4.5) {Node\\ n};
\node[align=center] at (2.5,4.5) {Node\\ j};
\node[align=center] at (2.15,8.2) {$e_i[t]$};
\node[align=center] at (1.4,7.7) {$u^g_i[t]$};
\node[align=center] at (3.35,7.3) {$u_i[t]$};

\end{tikzpicture}
\caption{The actuator node $i$ superimposes a private excitation whose realization is unknown to other nodes on to its control inputs}\label{actuator_probing}
\end{figure}
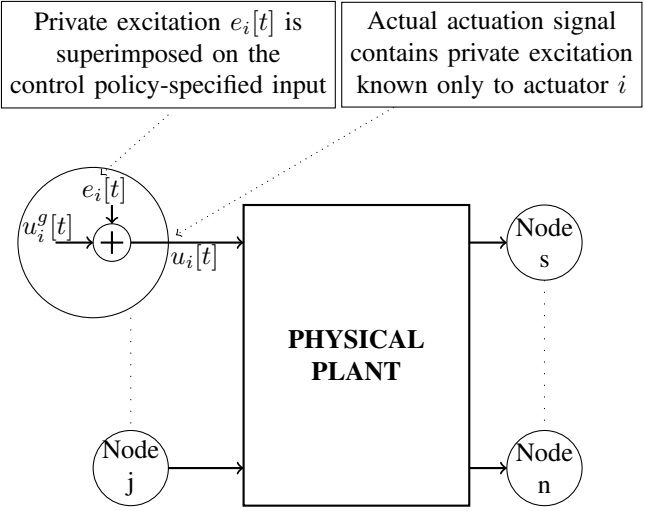


To see why private excitation helps, consider the example of a single-input-single-output (SISO) system where the sensor is malicious and the actuator is honest. Suppose that the control policy in place is $g = (g_1, g_2, \ldots , g_t, \ldots )$, where $g_t$ specifies the input to be applied at time $t$ in response to outputs up till that time. The actual outputs of the plant up to time $t$ are $x^t := (x[0], x[1], \ldots , x[t])$. However the outputs reported to it by the malicious sensor are $z^t := (z[0], z[1], \ldots , z[1])$, which may differ from $x^t$. The actuator therefore applies the control input $u[t] = g_t(z^t)$ at time $t$, without using any private excitation. Then, we have the system
\begin{align}
x[t+1]=ax[t]+bg_t(z^t)+w[t+1],
\end{align}
where $w[t]$ is the zero-mean process noise with variance $\sigma_w^2$.   

If the actuator does not superimpose a private excitation, the sensor knows, for each time $t$, the input $u[t]$ applied by the actuator. This is because it knows both the measurement sequence $z^t$ that it reported to the actuator as well as the control policy $\{g_1,g_2,...\}$ that the honest actuator has implemented. Hence, the malicious sensor can report a sequence of measurements $\{z[t]\}$ to the actuator \emph{without even "looking" at the output}, but by simply "simulating a linear stochastic system" after generating \emph{its own i.i.d. process noise $\{w'\}$} from the same distribution as that of $\{w\}$, as follows. 
\begin{align}
z(t+1)=az(t)+bg_t(z^t)+w'(t+1).
\end{align}
The actuator cannot detect that the sensor is malicious since the sequence $\{w'\}$, having been chosen from the same distribution as $\{w\}$, could have been the actual process noise. 

However, by superimposing a private excitation that is \emph{unknown} to the sensor, the actuator \emph{forces} the sensor to report measurements that are correlated with $\{e_i\}$, lest it be exposed, as we will show in the sequel. In the following sections, we prove that thereby constraining the sensor to report measurements that are correlated with the private excitation essentially limits the amount of distortion that the sensor can get away with while remaining undetected to be essentially zero in a mean-square sense. 

This active defense technique is similar in spirit to the technique of digital watermarking \cite{watermarking1} in electronic documents. Electronic documents can be easily transmitted and reproduced in large numbers. In doing so, the source of the document may be deleted, and can result in copyright violations. To protect the identity of its author, or any other information about the document, the electronic document is "watermarked" before being made available electronically. A digital watermark is a digital code that is robustly embedded in the original document \cite{watermarking2}. By robust, it is meant that the code cannot be destroyed without destroying the contents of the document. This code typically contains information about the document that needs to be preserved. Though preferable, it is not a requirement that the watermark be imperceptible. The only requirement is that it does not distort the actual contents beyond certain acceptable limits \cite{watermarking2}. Applications of digital watermarking also include data authentication, where fragile watermarks are used which get destroyed when the data is tampered with \cite{watermarking2}, and data monitoring and tracking.

This approach for secure control is analogous to digital watermarking. As we show in the subsequent sections, with regard to, for example, the above case where we considered a compromised sensor, injecting private excitation that is unknown to the sensor effectively "watermarks" the process noise, in the sense that the sensor cannot separate the private excitation from the process noise. Hence, any attempt on the part of the sensor to distort the process noise (which is the only component of the output unknown to the actuators) will also distort the "watermark," allowing the honest nodes to detect malicious activity. 

Specifically, given a general MIMO system of the form (\ref{firsteq}), we define $$\mathbf{v}[k]\coloneqq \mathbf{z}[k]-A\mathbf{z}[k-1]-B\mathbf{u}[k-1]-\mathbf{w}[k],$$ so that if the sensors truthfully report $\mathbf{z}[t]\equiv \mathbf{y}[t],$ then, $\mathbf{v}\equiv 0.$ As will be shown, the sequence $\{\mathbf{v}\}$ has the interpretation of the additive distortion introduced by the malicious sensors to the process noise. A similar definition can also be provided for a partially observed system, as we will show in Section \ref{partially_observed}. In this case, the sequence $\{\mathbf{v}\}$ has the interpretation of the additive distortion introduced by the malicious sensors to the innovations process. Now, based on $\{\mathbf{v}\}$, we define the following quantity.

\noindent\textbf{Definition:} (Additive distortion power of malicious sensors) We call $$\limsup_{T\to\infty} \frac{1}{T}\sum_{k=1}^{T} \|\mathbf{v}[k]\|^2$$ as the \emph{additive distortion power} of the malicious sensors.

The fundamental security guarantee provided by Dynamic Watermarking is that the additive distortion power is restricted to be zero if the malicious sensors are to remain undetected. We establish this in several linear control system contexts in this paper.

\section{Active Defense for Networked Cyber-Physical Systems: The SISO Case with Gaussian noise}\label{siso}
In this section, to illustrate the results in a simple context, we focus on single-input, single-output linear stochastic dynamical systems with Gaussian noise. The system is described by 
\begin{align}
x[t+1]=ax[t]+bu[t]+w[t+1], \label{siso_model}
\end{align}
where $a,b,x[t],u[t],w[t]\in\mathbb{R}$, with $\{w[t]\}$ being zero-mean i.i.d. Gaussian process noise of variance $\sigma_w^2$. The actuator wishes to implement a control law $\{g_t\}$, i.e., it wishes to implement $u[t]=g_t(x^t)$, where $x^t\coloneqq (x[0],x[1],\cdots,x[t])$. However, the actuator does not have access to $x^t$. It relies on a sensor that measures $x[t]$. However, since the sensor could be malicious, it reports measurements $z[t]$ to the actuator, where $z[t]$ could differ from $x[t]$. We consider an \emph{honest} actuator that is meant to, and implements the control policy $\{g\}$, but  adds a private excitation $\{e\}$ as a defense. Specifically, the actuator applies to the system the input
\begin{align}
u[t]=g_t(z^t)+e[t]. \label{siso_actuation}
\end{align}
Note that even though it implements the control policy $\{g\}$, the policy is applied to the measurements $z[t]$ reported by the sensor, which could differ from the true output $x[t]$. The private excitation $e[\cdot]$ added is independent and identically distributed (i.i.d.) and Gaussian of mean 0 and variance $\sigma_e^2$. 
Therefore, the system evolves in closed-loop as
\begin{align}
x[t+1]=ax[t]+bg_t(z^t)+be[t]+w[t+1].
\end{align}

We propose that the honest actuator perform certain "tests" to check if the sensor is malicious or not. Towards developing these tests, note that the actual sequence of states $\{x[t]\}$ of the system satisfies
\begin{align}
x[t+1]-ax[t]-bg_t(z^t)=be[t]+w[t+1].
\end{align}
Therefore, we have 
\begin{flalign}
&\{x[t+1]-ax[t]-bg_t(z^t)\}_t\sim \nonumber\\
&\;\;\;\;\;\;\;\;\;\;\;\;\;\;\;\;\;\;\;\;\;\;\;\;\;\;\;\;\;\;i.i.d.\; \mathcal{N}(0,b^2\sigma_e^2+\sigma_w^2),\label{test1_foundation}\\\text{and}\nonumber\\
&\{x[t+1]-ax[t]-bg_t(z^t)-be[t]\}_t\sim \nonumber\\
&\;\;\;\;\;\;\;\;\;\;\;\;\;\;\;\;\;\;\;\;\;\;\;\;\;\;\;\;\;\;\;\;\;\;\;\;\;\;\;\;\;\;\;\;\;\;\;\;\;i.i.d.\;\mathcal{N}(0,\sigma_w^2).\label{test2_foundation}
\end{flalign}

Based on the above observations, we propose that the actuator perform the following natural tests for variance to check if the sensor is honestly reporting $x[t]$. The actuator checks if the reported sequence $\{z[t]\}$ satisfies conditions (\ref{test1_foundation}) and (\ref{test2_foundation}), which the true output $\{x[t]\}$ would satisfy if the sensor were truthfully reporting $z[t]\equiv x[t]$. We write the tests in an asymptotic form below, as a test conducted over an infinite time interval. They can be reduced to statistical tests over a finite time interval in standard ways, which we elaborate on in Section-\ref{statTests}.  
\begin{itemize}
\item[1)] \textbf{Actuator Test 1:} Check if the reported sequence of measurements $\{z[t]\}$ satisfies
\begin{flalign}
\lim_{T\to\infty} \frac{1}{T} \sum_{k=0}^{T-1}(z[k+1]-az[k]-bg_k(z^k)-be[k])^2=\sigma_w^2.\label{test2}
\end{flalign}

\item[2)] \textbf{Actuator Test 2:} Check if the reported sequence of measurements $\{z[t]\}$ satisfies
\begin{flalign}
&\lim_{T\to\infty} \frac{1}{T} \sum_{k=0}^{T-1} (z[k+1]-az[k]-bg_k(z^k))^2\nonumber\\ 
&=(b^2\sigma_e^2+\sigma_w^2).\label{test1}
\end{flalign}
\end{itemize}

Define $$v[t+1]\coloneqq z[t+1]-az[t]-bg_t(z^t)-be[t]-w[t+1],$$ so that for an honest sensor which reports $z[t]\equiv x[t]$, $v[t]=0\;\; \forall t$. We term the quantity $$\lim_{T\to\infty} \frac{1}{T} \sum_{k=1}^{T} v^2[k]$$ the \emph{additive distortion power} of a malicious sensor for reasons explained later. The ensuing theorem proves that a malicious sensor with only zero additive distortion power can pass the above two tests, thereby remaining undetected.  

\begin{theorem}
If $\{z[t]\}$ passes tests (\ref{test2}) and (\ref{test1}), thereby remaining undetected, then, 
\begin{equation}
\lim_{T\to\infty} \frac{1}{T} \sum_{k=1}^{T} v^2[k]=0.\label{thm1a} 
\end{equation}

\end{theorem}
\begin{proof}
Since $\{z\}$ satisfies (\ref{test2}), we have,
\begin{align}
&\lim_{T\to\infty} \frac{1}{T} \sum_{k=1}^{T} (v[k]+w[k])^2=\sigma_w^2. \label{t1use1}
\end{align}
Hence,
\begin{align}
&\lim_{T\to\infty} \frac{1}{T} \sum_{k=1}^{T} v^2[k]+2v[k]w[k]+w^2[k]=\sigma_w^2. \nonumber
\end{align}
Since $\lim_{T\to\infty} \frac{1}{T} \sum_{k=1}^{T} w^2[k]=E\{w^2[k]\}=\sigma_w^2$, we have
\begin{align}
\lim_{T\to\infty} \frac{1}{T} \sum_{k=1}^{T} v^2[k]+ \lim_{T\to\infty} \frac{1}{T} \sum_{k=1}^{T} 2v[k]w[k]=0. \label{t1use2}
\end{align}
Since $\{z\}$ also satisfies (\ref{test1}), we have,
\begin{align*}
&\lim_{T\to\infty} \frac{1}{T} \sum_{k=1}^{T} (v[k]+be[k-1]+w[k])^2=b^2\sigma_e^2+\sigma_w^2.
\end{align*}
So,
\begin{align*}
\lim_{T\to\infty} \frac{1}{T} \sum_{k=1}^{T} (v[k]+w[k])^2+ \lim_{T\to\infty} \frac{1}{T} \sum_{k=1}^{T} b^2e^2[k-1] \\+\lim_{T\to\infty} \frac{1}{T} \sum_{k=1}^{T} 2be[k-1](v[k]+w[k])=b^2\sigma_e^2+\sigma_w^2.
\end{align*}
Using (\ref{t1use1}), and the fact that $\{e\}$ has variance $\sigma_e^2$, the above reduces to
\begin{align*}
\lim_{T\to\infty} \frac{1}{T} \sum_{k=1}^{T} e[k-1](v[k]+w[k])=0.
\end{align*}
Invoking the fact that $e[k-1]$ and $w[k]$ are independent $\forall k$ and zero mean, it further reduces to
\begin{align}
\lim_{T\to\infty} \frac{1}{T} \sum_{k=1}^{T} e[k-1]v[k]=0. \label{key1}
\end{align}
The above equation implies that the sequence $\{v\}$, added by the sensor, must be empirically uncorrelated with the actuator's private noise sequence $\{e\}$. 

Let $\mathcal{S}_k\coloneqq \sigma (x^k,z^{k},e^{k-2})$, $\widehat{w}[k]\coloneqq E[w[k]\big | \mathcal{S}_k]$. Since $$w[k]=x[k]-ax[k-1]-bg_{k-1}(z^{k-1})-be[k-1],$$ we have $$(x^{k-2},e^{k-2})\to (x[k-1],x[k],z^{k})\to w[k]$$
forming a Markov chain. Consequently, $\widehat{w}[k]\coloneqq E[w[k]\big |\sigma(e^{k-2},x^{k-2},x[k-1],x[k],z^{k})]=E[w[k]\big |\sigma(x[k-1],x[k],z^{k})].$ Since $x[k]-ax[k-1]-bg_{k-1}(z^{k-1})$ (which is equal to $be[k-1]+w[k]$) is i.i.d. Gaussian for different $k$, we have \cite{KumarVaraiya}
\begin{equation}
\widehat{w}[k]=\frac{\sigma_w^2}{b^2\sigma_e^2+\sigma_w^2}(be[k-1]+w[k])=\beta (be[k-1]+w[k]),\label{use}
\end{equation}
where $\beta\coloneqq \frac{\sigma_w^2}{b^2\sigma_e^2+\sigma_w^2}<1$. 

Let $\widetilde{w}[k]\coloneqq w[k]-\widehat{w}[k]$. Then, $(\widetilde{w}[k-1],\mathcal{S}_k)$ is a Martingale difference sequence. This is because $\widetilde{w}[k-1]\in\mathcal{S}_k$, and
\begin{equation}
E[\widetilde{w}[k]\;\big | \mathcal{S}_{k}]=0.
\end{equation}
We also have $v[k]\in\mathcal{S}_k$ (in fact, $v[k]\in\sigma(x^k,z^{k})$). Hence, Martingale Stability Theorem (MST) \cite{MST} applies, and we have
\begin{equation}
\sum_{k=1}^{T} v[k]\widetilde{w}[k]=o(\sum_{k=1}^{T} v^2[k]) + O(1).
\end{equation}
Now,
\begin{align}
\sum_{k=1}^{T} v[k]w[k]&=\sum_{k=1}^{T} v[k](\widehat{w}[k]+\widetilde{w}[k])\nonumber\\
&=\sum_{k=1}^{T} v[k]\widehat{w}[k]+o(\sum_{k=1}^{T} v^2[k]) + O(1).\nonumber
\end{align}
Employing the specific form of the estimate (\ref{use}), we have from the above,
\begin{align}
\sum_{k=1}^{T}v[k]w[k]=&\beta b\sum_{k=1}^{T} v[k]e[k-1] +\beta \sum_{k=1}^{T} v[k]w[k] \nonumber\\
&+ o(\sum_{k=1}^{T} v^2[k]) + O(1).\nonumber
\end{align}
Hence,
\begin{align}
\sum_{k=1}^{T}v[k]w[k]=&\frac{\beta b}{1-\beta} \sum_{k=1}^{T} v[k]e[k-1]+o(\sum_{k=1}^{T} v^2[k])+ O(1).\nonumber
\end{align}
From (\ref{key1}), we have $\sum_{k=1}^{T} v[k]e[k-1]=o(T)$. It follows that
\begin{align}
&\sum_{k=1}^{T} v[k]w[k]=o(\sum_{k=1}^{T} v^2[k])+o(T)+O(1).
\end{align}
So,
\begin{align}
\sum_{k=1}^{T} v^2[k]+ \sum_{k=1}^{T} 2v[k]w[k]=(1+o(1))(\sum_{k=1}^{T} v^2[k]) &+ o(T) \nonumber\\
&+ O(1)\nonumber
\end{align}
Dividing the above equation by $T$, taking the limit as $T\to\infty$, and invoking (\ref{t1use2}) completes the proof. 
\end{proof}

\noindent\textbf{Remark:} Note that the only sources of uncertainty in the system are the initial state of the system $x[0]$ and the sequence of noise realizations $\{w[1],w[2],w[3],\cdots\}$. The sensor reporting a sequence of measurements is equivalent to it reporting a sequence of process noise realizations, since the actuator expects $z[t+1]-az[t]-bg_t(z^t)-be[t]$, which it can compute, to be equal to the process noise $w[t+1]$. From the definition of $v[t]$, we have 
\begin{equation*}
z[t+1]-az[t]-bg_t(z^t)-be[t]=w[t+1]+v[t+1].
\end{equation*}
The left hand side of the above equation can be computed by the actuator, and therefore, it can also compute the sequence $\{w+v\}$. What the theorem states is that a malicious sensor cannot distort the noise realization $\{w[1],w[2],w[3],\cdots\}$ beyond adding a zero-power sequence to it. Suppose the control law $g$ has been designed to provide good noise regulation performance of a stable system, then the performance of the system is good, subject only to the slightly increased cost of the private excitation of variance $\sigma_e^2$. As we discuss in Section \ref{statTests}, this can be reduced to a low enough value that permits detection of a malicious sensor within a specified delay with an acceptable level of false alarm probability. 

\begin{theorem} Suppose $|a|<1$, i.e., the system is stable.  
\begin{itemize}
\item[(i)] Define the distortion $d[t]\coloneqq z[t]-x[t]$. Then, 
\begin{equation*}
\lim_{T\to\infty} \frac{1}{T}\sum_{k=0}^{T-1} d^2[k]=0.
\end{equation*}

\item[(ii)] If the malicious sensor is to remain undetected, the mean-square performance of $x[t]$ is the same as the reported mean-square performance $z[t]$ that the actuator believes it is:
$$\lim_{T \to \infty} \frac{1}{T} \sum_{k=0}^{T-1} x^2[k] = \lim_{T \to \infty} \frac{1}{T} \sum_{k=0}^{T-1} z^2[k].$$  

\item[(iii)] Suppose the control law is $u^g(t) = fx(t)$ with $|a+bf| < 1$. The malicious sensor cannot compromise the performance of the system if it is to remain undetected, i.e., the mean-square performance of the system is $$ \lim_{T \to \infty} \frac{1}{T} \sum_{k=0}^{T-1} x^2[k] = \frac{ \sigma_w^2 + b^2 \sigma_e^2 }{1 - |a+bf|^2 }.$$
\end{itemize}
\end{theorem}
\begin{proof}
Note that
\begin{flalign*}
&d[k+1]\\
&=z[k+1]-x[k+1]\\
&=(az[k]+bg_k(z^k)+be[k]+w[k+1]+v[k+1])\\
&\;\;\;\;\;-(ax[k]+bg_k(z^k)+be[k]+w[k+1])\\
&=a(z[k]-x[k])+v[k+1]\\
&=ad[k]+v[k+1].
\end{flalign*}
The distortion experienced by the actuator can therefore be thought of as the output of a linear dynamical system driven by an input sequence $\{v[t]\}$ satisfying (\ref{thm1a}). Therefore,
\begin{flalign*}
&d[k]=\sum_{n=0}^{k-1} a^nv[k-n],
\end{flalign*}
where $\lim_{T\to\infty} \frac{1}{T} \sum_{p=1}^{T} v^2[p]=0.$

From the stability of $a$, it follows that
\begin{equation}
\lim_{T\to\infty} \frac{1}{T}\sum_{k=0}^{T-1} d^2[k]=0.
\end{equation}

Note that $x[k] = z[k] - d[k]$. Hence $x^2[k] = z^2[k] + d^2[k] - 2 (\gamma z[k])(\gamma^{-1} d[k])$. 
Now $|2 (\gamma z[k])(\gamma^{-1} d[k] )| \leq (\gamma^2 z^2[k]) + (\gamma^{-2} d^2[k] )$. Therefore, 
\begin{flalign*}
\lim_{T \to \infty} \frac{1}{T} \sum_{k=0}^{T-1} x^2[k] \leq &\lim_{T \to \infty} \frac{1}{T} \sum_{k=0}^{T-1} (1 + \gamma^2) z^2[k] \nonumber\\
&+ (1 + \gamma^{-2}) d^2[k].
\end{flalign*}
Since the result is true for any $\gamma >0$, taking the limit $\gamma \to 0$ and noting that the mean-square of $d$ is 0 gives
\begin{flalign}
\lim_{T \to \infty} \frac{1}{T} \sum_{k=0}^{T-1} x^2[k] \leq \lim_{T \to \infty} \frac{1}{T} \sum_{k=0}^{T-1} z^2[k].\label{t2bu1}
\end{flalign}
Similarly, since $z[k]=x[k]+d[k]$, we have $z^2[k]=x^2[k]+d^2[k]+2(\gamma x[k]) (\gamma^{-1}z[k]).$ Repeating the same argument as before, we have
\begin{flalign*}
\lim_{T \to \infty} \frac{1}{T} \sum_{k=0}^{T-1} z^2[k] \leq &\lim_{T \to \infty} \frac{1}{T} \sum_{k=0}^{T-1} (1 + \gamma^2) x^2[k] \nonumber\\
&+ (1 + \gamma^{-2}) d^2[k].
\end{flalign*}
Taking the limit as $\gamma\to 0$, and noting that mean-square of $d$ is 0, the above reduces to 
\begin{flalign}
\lim_{T \to \infty} \frac{1}{T} \sum_{k=0}^{T-1} z^2[k] \leq \lim_{T \to \infty} \frac{1}{T} \sum_{k=0}^{T-1} x^2[k].\label{t2bu2}
\end{flalign}
The second result follows from (\ref{t2bu1}) and (\ref{t2bu2}). 
The third result is immediate noting that the mean-square of $\{z\}$ converges to $\frac{ \sigma_w^2 + b^2 \sigma_e^2 }{1 - |a+bf|^2 }$.
\end{proof}

\section{Active Defense for Networked Cyber-Physical Systems: The SISO ARX Case}\label{sisoarx}
The results developed in the previous section can be extended to a more general ARX system model. Specifically, we consider a unit delay, strictly minimum phase, single-input, single-output system described by
\begin{align}
y[t+1]=-\sum_{m=0}^{p}a_my[t-m]+\sum_{r=0}^{h}b_ru[t-r]+w[t+1], \label{siso_arx_model}
\end{align}
where $a_m,b_r,y[t],u[t],w[t]\in\mathbb{R}$, $b_0\neq 0$, and with $\{w[t]\}$ being zero-mean i.i.d. Gaussian process noise of variance $\sigma_w^2$. Let $q^{-1}$ denote the backward shift operator. The above system can equivalently be expressed as 
\begin{align}
A(q^{-1})y[t]=q^{-1}B(q^{-1})u[t]+w[t],
\end{align}
where $A(q^{-1})\coloneqq 1+a_0q^{-1}+a_1q^{-2}+...+a_pq^{-(p+1)},$ and $B(q^{-1})\coloneqq b_0+b_1q^{-1}+...+b_hq^{-h},$ with $B(q^{-1})$ being strictly minimum phase, i.e., all its roots lie strictly outside the unit circle.  

We consider an honest actuator that is meant to, and implements the control policy $\{g\}$, and adds a private excitation $\{e\}$ as a defense. Unlike in the system considered in Section-\ref{siso}, the output of the ARX system at any particular time depends on the past inputs. Specifically, the output at any time instant contains contributions of private excitation injected in the past. Hence, simply injecting an i.i.d. sequence of Gaussian random variables will not result in an output distribution that is i.i.d. across time. 

However, since the actuator knows the past values of the private excitation, and also the transfer function of the system, it can perform "pre-equalization" by filtering the private excitation sequence before injecting it into the system. The filter, which we refer to as the pre-equalizer, has to be chosen in such a way that the component of the private excitation that appears in the output of the plant is an i.i.d. sequence of mean 0 and variance $b_0^2\sigma_e^2$. We let $$e'[t] \coloneqq - \frac{1}{b_0} (b_1 e'[t-1] + b_2 e'[t-2] + \ldots + b_h e'[t-h] ) + e[t],$$ where $e[t]$ is i.i.d. and Gaussian of mean 0 and variance $\sigma_e^2$. Since the system is strictly minimum phase, this is a stable generation of an excitation sequence \cite{KumarVaraiya}.

With a pre-equalizer in place, the effective input applied by the actuator is
\begin{align}
u[t]=g_t(z^t)+e'[t].
\end{align}
where $\{e'\}_t$ is the output of the pre-equalizer, and $z^t$ is the sequence of measurements reported by the sensor up to time $t$. Therefore, the system evolves in closed-loop as
\begin{align}
y[t+1]=-\sum_{m=0}^{p}a_my[t-m]+\sum_{r=0}^{h}b_rg_{t-r}(z^{t-r})\nonumber\\
+b_0e[t]+w[t+1],
\end{align}
where $\{e[t]\}$ is a sequence of i.i.d. Gaussian random variables with mean 0 and variance $\sigma_e^2$. Hence, from the point of view of actuator tests and associated results, the above problem is similar to that described in Section-\ref{siso}.

We propose that the actuator perform the following tests to check if the sensor is reporting the measurements honestly. 

\begin{itemize}
\item[1)] \textbf{Actuator Test 1:} Check if the reported sequence of measurements $\{z[t]\}$ satisfies
\begin{flalign}
\lim_{T\to\infty} \frac{1}{T} \sum_{k=0}^{T-1}(&z[k+1]+\sum_{m=0}^{p}a_mz[k-m]-\nonumber\\
&\sum_{r=0}^{h} b_rg_{k-r}(z^{k-r})-b_0e[k])^2=\sigma_w^2.\label{test2_arx}
\end{flalign}

\item[2)] \textbf{Actuator Test 2:} Check if the reported sequence of measurements $\{z[t]\}$ satisfies
\begin{flalign}
\lim_{T\to\infty} \frac{1}{T} \sum_{k=0}^{T-1}(&z[k+1]+\sum_{m=0}^{p}a_mz[k-m]-\nonumber\\
&\sum_{r=0}^{h} b_rg_{k-r}(z^{k-r}))^2=(b_0^2\sigma_e^2+\sigma_w^2).\label{test1_arx}
\end{flalign}
\end{itemize}
Consequently, Theorem 1 and Theorem 2 can be extended as follows.  
\begin{theorem}
\begin{itemize}
\item[1)] Let $v[t+1]\coloneqq z[t+1]+\sum_{m=0}^{p}a_mz[k-m]-\sum_{r=0}^{h} b_rg_{k-r}(z^{k-r})-b_0e[t]-w[t+1]$, so that for an honest sensor which reports $z[t]\equiv y[t]$, $v[t]=0\;\; \forall t$. If $\{z[t]\}$ passes tests (\ref{test2_arx}) and (\ref{test1_arx}), then, 
\begin{equation}
\lim_{T\to\infty} \frac{1}{T} \sum_{k=1}^{T} v^2[k]=0.\label{thm1a_arx} 
\end{equation}
\item[2)] If the malicious sensor is to remain undetected, the mean-square performance of $\{y\}$ is what the actuator believes it is, which is the mean-square performance of $\{z\}$. Hence if the control law was designed to provide a certain mean-square performance, then that is the value that is indeed attained.
\end{itemize}

\end{theorem}
\begin{proof}
Omitted, since the proof follows the same sequence of arguments as the proof of Theorem 1. 
\end{proof}

In the following section, we extend the technique to systems with arbitrary delay and colored noise, i.e., the more general ARMAX model, see \cite{ast70}.

\section{Active Defense for Networked Cyber-Physical Systems: The SISO ARMAX Case}\label{ARMAX}
A general ARMAX system with arbitrary but finite delay is a model that is encountered often in process control. In this section, we develop a Dynamic Watermarking scheme to secure such systems. We show that by commanding the actuator to inject private excitation whose spectrum is matched to that of the colored process noise entering the system, it can be ensured that a malicious sensor is constrained to distorting the process noise, the only quantity unknown to the actuator, by at most a zero-power signal. This amounts to a form of Internal Model Principle \cite{internal_model} for Dynamic Watermarking, and is a phenomenon that doesn't emerge in the analysis of a simple SISO or an ARX system treated in the previous sections. 

A general ARMAX system with finite delay $l$ is described by
\begin{align}
y[t]=-\sum_{k=1}^{p} a_ky[t-k]+\sum_{k=0}^{h}b_ku[t-l-k]\nonumber\\
+\sum_{k=0}^{r}c_kw[t-k].\label{ARMAX_model}
\end{align}
Without loss of generality, we assume that $c_0=1$. Let $q^{-1}$ be the backward shift operator, so that $q^{-1}y[t]=y[t-1].$ Then, the above system can be expressed as
\begin{align}
A(q^{-1})y[t]=q^{-l}B(q^{-1})u[t]+C(q^{-1})w[t],
\end{align}
where $a_0\coloneqq 1$, $A(q^{-1})\coloneqq a_0+a_1q^{-1}+...+a_pq^{-p}$, $B(q^{-1})\coloneqq b_0+b_1q^{-1}+...+b_hq^{-h}$, $C(q^{-1})\coloneqq c_0+c_1q^{-1}+...+c_rq^{-r}$, and $B(q^{-1})$ and $C(q^{-1})$ are assumed to be strictly minimum phase. To secure the system, the actuator applies the control
\begin{align}
u[t]=u^g[t]+B^{-1}(q^{-1})C(q^{-1})e[t],\label{input_ARMAX}
\end{align}
where $u^g[t]$ is the control policy-specified input sequence and $\{e\}$ is a sequence of i.i.d Gaussian random variables of zero mean and variance $\sigma_e^2$, denoting the actuator node's private excitation. Since $B(q^{-1})$ is assumed to be minimum phase, (\ref{input_ARMAX}) is a stable generation of the control input $u[t]$. Consequently, the output of the system obeys
\begin{align}
A(q^{-1})y[t]=q^{-l}B(q^{-1})u^g[t]+q^{-l}C(q^{-1})e[t]\nonumber\\
+C(q^{-1})w[t],
\end{align}
implying that
\begin{align}
y[t]=-\sum_{k=1}^{p} a_ky[t-k]+\sum_{k=0}^{h}b_ku^g[t-l-k]\nonumber\\
+\sum_{k=0}^{r}c_ke[t-l-k]+\sum_{k=0}^{r}c_kw[t-k].
\end{align}
While the $q-$domain derivation is just formal, the above difference equation can be obtained rigorously.  

Define $\lambda[t]\coloneqq e[t-l]+w[t]$. Then, the output of the system can be expressed as
\begin{align}
y[t]=-\sum_{k=1}^{p} a_ky[t-k]+\sum_{k=0}^{h}b_ku^g[t-l-k]\nonumber\\
+\sum_{k=0}^{r}c_k\lambda[t-k].
\end{align}
We now develop tests that the actuator should perform to detect maliciousness of the sensor. The fundamental idea behind the tests is to check if the prediction-error of the reported sequence $\{z\}$ has the appropriate statistics. Specifically, the actuator processes the reported measurements with a prediction-error filter defined through the following recursion for all $t\geq 0$.
\begin{align}
{z}_{t|t-1}=&-\sum_{k=1}^{p}a_k{z}[t-k]+\sum_{k=0}^{h}b_ku^g[t-l-k]\nonumber\\
&+\sum_{k=1}^{r}c_k\widetilde{z}[t-k],\label{pred_error_1}\\
&\widetilde{z}[t]=z[t]-z_{t|t-1}.\label{pred_error_2}
\end{align}
The filter is assumed to be initialized with $\widetilde{z}[-k]=\lambda[-k], k\in\{-1,-2,...,-r\}$. It is easy to verify that the above filter produces $\widetilde{z}[t]\equiv \lambda[t]$ if $z[t]\equiv y[t]$. Based on this observation, we propose that the actuator perform the following tests.

\begin{enumerate}
\item \textbf{Actuator Test 1:} The actuator checks if
\begin{align}
\lim_{T\to\infty} \frac{1}{T} \sum_{k=0}^{T-1} (\widetilde{z}[k]-e[k-l])^2=\sigma_w^2.\label{ARMAX_test}
\end{align}
\item \textbf{Actuator Test 2:} The actuator checks if
\begin{align}
\lim_{T\to\infty} \frac{1}{T} \sum_{k=0}^{T-1} \widetilde{z}^2[k]=\sigma_w^2+\sigma_e^2.\label{ARMAX_test2}
\end{align}
\end{enumerate}
The following theorem shows that the above tests suffice to ensure that a malicious sensor cannot introduce any distortion beyond addition of a zero-power signal to the process noise. i.e., a malicious sensor of only zero additive distortion power can pass the above tests to remain undetected. 

\begin{theorem}
Define $v[t]\coloneqq \sum_{k=0}^{p}a_kz[t-k]-\sum_{k=0}^{h}b_ku^g[t-l-k]-\sum_{k=0}^{r}c_k\lambda[t-k]$, so that for an honest sensor reporting $z\equiv y$, $v\equiv 0$. If the sensor passes tests (\ref{ARMAX_test}) and (\ref{ARMAX_test2}), then, $$\lim_{T\to\infty} \frac{1}{T}\sum_{k=0}^{T-1}v^2[k]=0.$$
\end{theorem}
\begin{proof}
From (\ref{pred_error_1}) and (\ref{pred_error_2}), one gets
\begin{align}
\sum_{k=0}^{r}c_k\widetilde{z}[t-k]&=z[t]+\sum_{k=1}^{p}a_k{z}[t-k]-\sum_{k=0}^{h}b_ku^g[t-l-k]\nonumber\\
&=\sum_{k=0}^{r}c_k\lambda[t-k]+v[t].
\end{align}
Equivalently, assuming appropriate initial conditions, one has
\begin{align}
C(q^{-1})\widetilde{z}[t]&=C(q^{-1})\lambda[t]+v[t].\nonumber
\end{align}
This gives
\begin{align}
\widetilde{z}[t]&=\lambda[t]+v_f[t],\nonumber
\end{align}
where $v_f[t]\coloneqq C^{-1}(q^{-1})v[t]$. Now,
\begin{align}
\widetilde{z}[t]&=\lambda[t]+v_f[t]\nonumber\\
&=w[t]+e[t-l]+v_f[t].
\end{align}

Since $\{\widetilde{z}\}$ passes (\ref{ARMAX_test}), we have from the above, 
\begin{align}
\lim_{T\to\infty} \frac{1}{T} \sum_{k=0}^{T-1} (v_f[k]+w[k])^2=\sigma_w^2.
\end{align}
This gives
\begin{align}
\lim_{T\to\infty} \frac{1}{T} \sum_{k=0}^{T-1} v_f^2[k]+2v_f[k]w[k]=0.\label{armaxuse1}
\end{align}
Since $\{\widetilde{z}\}$ also passes (\ref{ARMAX_test2}), we have
\begin{align}
\lim_{T\to\infty} \frac{1}{T} \sum_{k=0}^{T-1} (v_f[k]+w[k]+e[k-l])^2=\sigma_w^2+\sigma_e^2.
\end{align}
This gives
\begin{align*}
\lim_{T\to\infty} \frac{1}{T} \sum_{k=0}^{T-1} v_f^2[k]+2v_f[k]w[k]+2v_f[k]e[k-l]=0.
\end{align*}
Combining the above with (\ref{armaxuse1}), we have
\begin{align}
\lim_{T\to\infty} \frac{1}{T} \sum_{k=0}^{T-1} v_f[k]e[k-l]=0.\label{armaxuse2}
\end{align}

Comparing (\ref{armaxuse1}) and (\ref{armaxuse2}) with (\ref{t1use2}) and (\ref{key1}), we see that the filtered distortion measure $\{v_f\}$ behaves the same way $\{v\}$ does in the white noise case. Proceeding the same way as in the proof of Theorem 1, one arrives at
\begin{align}
\lim_{T\to\infty} \frac{1}{T} \sum_{k=0}^{T-1} v_f^2[k]=0.
\end{align}
Since $v[k]=C(q^{-1})v_f[k]$, and $C(q^{-1})$ is minimum phase, it follows that $\lim_{T\to\infty} \frac{1}{T} \sum_{k=0}^{T-1} v^2[k]=0.$
\end{proof}

\section{Active Defense for Networked Cyber-Physical Systems: SISO Systems with Partial Observations}\label{partially_observed}
In this section, we address SISO systems with noisy, partial observations. We consider a $p^{th}$ order single input single output system described by
\begin{align}
\mathbf{x}[t+1]=A\mathbf{x}[t]+Bu[t]+\mathbf{w}[t+1],\\
y[t+1]=C\mathbf{x}[t+1]+n[t+1].
\end{align}
where $\mathbf{x}[t]\in\mathbb{R}^{p},$ and $A$, $B$, and $C$ are known matrices of appropriate dimensions. The actuator, being honest, applies the input
\begin{align}
u[t]=g_t({z}^t)+e[t],
\end{align}
where $g_t({z}^t)$ is the control policy-specified input, and $e[t]\sim\mathcal{N}(0,\sigma_e^2)$ is a sequence of i.i.d random variables denoting the actuator node's private excitation. Consequently, the system evolves as
\begin{align}
\mathbf{x}[t+1]=A\mathbf{x}[t]+Bg_t({z}^t)+be[t]+\mathbf{w}[t+1],\\
y[t+1]=C\mathbf{x}[t+1]+n[t+1].
\end{align}
Let $z[t]$ be the measurement reported by the sensor at time $t$. Since the sensor can be malicious, $z[t]$ need not equal $y[t]$ for every $t$. 

The actuator performs Kalman filtering on the reported measurements $\{z\}$ as follows. 
\begin{align}
\mathbf{\widehat{x}}_F(k+1|k)=A\mathbf{\widehat{x}}_F(k|k)+Bg_k({z}^k)+Be[k],
\end{align}
\begin{align}
\mathbf{\widehat{x}}_F(k+1|k+1)=A\mathbf{\widehat{x}}_F(k|k)+Bg_k({z}^k)+Be[k]\nonumber\\
+K_k\nu_F[k+1],
\end{align}
where $\nu_F[k+1]\coloneqq z[k+1]-C\mathbf{\widehat{x}}_F(k+1|k),$ denotes the (possibly) faulty innovations computed by the actuator, and $K_k$ is the Kalman gain at time $k$. 

We also define a Kalman filter that operates on the true measurements $y[t]$ as follows.
\begin{align}
\mathbf{\widehat{x}}_R(k+1|k)=A\mathbf{\widehat{x}}_R(k|k)+Bg_k({z}^k)+Be[k],
\end{align}
\begin{align}
\mathbf{\widehat{x}}_R(k+1|k+1)=A\mathbf{\widehat{x}}_R(k|k)+Bg_k({z}^k)+Be[k]\nonumber\\
+K_k\nu_R[k+1],\label{trueKF}
\end{align}
where $\nu_R[k+1]\coloneqq y[k+1]-C\mathbf{\widehat{x}}_R(k+1|k)$ is the "real innovations" or "true innovations" of the system. Of course, the actuator cannot implement the latter Kalman filter since it may not receive $y[k]$ from the sensor. 

We now define
\begin{align}
\mathbf{v}[k+1]\coloneqq \mathbf{\widehat{x}}_F(k+1|k+1)-A\mathbf{\widehat{x}}_F(k|k)-Bg_k({z}^k)\nonumber\\
-Be[k]-K_k\nu_R[k+1],\label{definev}
\end{align}
so that for an honest sensor which reports $z\equiv y$, we have $\nu\equiv \nu_R$ and consequently, $\mathbf{v}\equiv 0$. The actuator performs the following two tests to detect maliciousness of the sensor.
\begin{itemize}
\item[1)] \textbf{Actuator Test 1:} Actuator checks if
\begin{align}
\lim_{T\to\infty} \frac{1}{T} \sum_{k=0}^{T-1} e[k](\mathbf{\widehat{x}}_F(k+1|k+1)-A\mathbf{\widehat{x}}_F(k|k)\nonumber\\
-Bg_k({z}^{k})-Be[k])=0\label{potest1}
\end{align}
\item[2)] \textbf{Actuator Test 2:} Actuator checks if
\begin{align}
&\lim_{T\to\infty} \frac{1}{T} \sum_{k=0}^{T-1} \nonumber\\
&(\mathbf{\widehat{x}}_F(k+1|k+1)-A\mathbf{\widehat{x}}_F(k|k)-Bg_k({z}^k)-Be[k])\nonumber\\
&(\mathbf{\widehat{x}}_F(k+1|k+1)-A\mathbf{\widehat{x}}_F(k|k)-Bg_k({z}^k)-Be[k])^T\nonumber\\
&=\sigma_R^2KK^T\label{potest2},
\end{align}
\end{itemize}
where $\sigma_R^2$ denotes the variance of the true innovations process, and $K$ is the steady-state Kalman gain. We assume that $(A,C)$ is observable so that the algebraic Riccati equation associated with the Kalman gain has a unique nonnegative definite solution \cite{KumarVaraiya}. The following theorem shows that the above tests suffice to ensure that a malicious sensor of only zero additive distortion power can pass the tests to remain undetected.
\begin{theorem}
Suppose that the reported sequence of measurements passes the tests (\ref{potest1}) and (\ref{potest2}). Then, 
\begin{enumerate}
\item \begin{equation}
\lim_{T\to\infty}\frac{1}{T}\sum_{k=0}^{T-1} \|\mathbf{v}[k+1]\|^2=0.
\end{equation}
\end{enumerate}
\end{theorem}
\begin{proof}
Since the reported sequence of measurements passes (\ref{potest1}), we have
\begin{align}
\lim_{T\to\infty} \frac{1}{T} \sum_{k=0}^{T-1} e[k] (K_k\nu_R[k+1]+\mathbf{v}[k+1])=0.\nonumber
\end{align}
Since the true innovations are independent of the zero-mean private excitation, it follows that
\begin{align}
\lim_{T\to\infty} \frac{1}{T} \sum_{k=0}^{T-1} e[k]\mathbf{v}[k+1]=0.\label{po1}
\end{align}
Since the reported sequence of measurements also passes (\ref{potest2}), we have
\begin{align}
\lim_{T\to\infty} \frac{1}{T} \sum_{k=0}^{T-1} &(K_k\nu_R[k+1]+\mathbf{v}[k+1])\nonumber\\
&(K_k\nu_R[k+1]+\mathbf{v}[k+1])^T=\sigma_R^2KK^T.\nonumber
\end{align}
Simplifying, the above gives
\begin{align}
&\lim_{T\to\infty} \frac{1}{T} \sum_{k=0}^{T-1} (K_k\nu_R[k+1]\mathbf{v}^T[k+1])\nonumber\\
&+(K_k\nu_R[k+1]\mathbf{v}^T[k+1])^T+(\mathbf{v}[k+1]\mathbf{v}^T[k+1])=0.\label{po2}
\end{align}

We define $\mathcal{S}_k\coloneqq \sigma(\mathbf{\widehat{x}}_R^{k|k},\mathbf{\widehat{x}}_F^{k|k},z^{k},e^{k-2},y^{k-1}),$ where $\mathbf{\widehat{x}}_R^{k|k}\coloneqq \{\mathbf{\widehat{x}}_R(k|k),\mathbf{\widehat{x}}_R(k-1|k-1),...,\mathbf{\widehat{x}}_R(-1|-1)\},$ and $\mathbf{\widehat{x}}_F^{k|k}$ is defined likewise. Define $\widehat{\nu}_R[k]\coloneqq E[\nu_R[k]\big | \mathcal{S}_k].$ From the Kalman filtering equations, we have
\begin{align*}
K_{k-1}\nu_R[k]=\mathbf{\widehat{x}}_R(k|k)-A\mathbf{\widehat{x}}_R(k-1|k-1)\nonumber\\
-Bg_{k-1}(z^{k-1})-Be[k-1]
\end{align*}
implying that 
\begin{align*}
(\mathbf{\widehat{x}}_R^{k-2|k-2},\mathbf{\widehat{x}}_F^{k|k},e^{k-2},y^{k-1})\to\nonumber\\(\mathbf{\widehat{x}}_R(k-1|k-1),\mathbf{\widehat{x}}_R(k|k),z^k)\to\nu_R[k]
\end{align*}
forms a Markov chain. Therefore, $\widehat{\nu}_R[k]\coloneqq E[\nu_R[k]\big |\sigma(\mathbf{\widehat{x}}_R^{k|k},\mathbf{\widehat{x}}_F^{k|k},z^{k},e^{k-2},y^{k-1})]=E[\nu_R[k]\big | \sigma(\mathbf{\widehat{x}}_R(k-1|k-1),\mathbf{\widehat{x}}_R(k|k),z^k)].$ Therefore, we have 
\begin{align}
K_{k-1}\widehat{\nu}_R[k]=K_{\nu}(K_{k-1}\nu_R[k]+Be[k-1]),\label{innov_mix_e}
\end{align}
where $K_\nu\coloneqq \sigma_R^2K_{k-1}K_{k-1}^T(\sigma_R^2K_{k-1}K_{k-1}^T+\sigma_e^2BB^T)^{-1}$. Define $K_{k-1}\widetilde{\nu}_R[k]\coloneqq K_{k-1}\nu_R[k]-K_{k-1}\widehat{\nu}_R[k].$ Then, $(K_{k-1}\widetilde{\nu}_R[k-1],\mathcal{S}_k)$ is a martingale difference sequence. This is because, $\widetilde{\nu}_R[k-1]\in\mathcal{S}_k,$ and $$E[\widetilde{\nu}_R[k]\big |\mathcal{S}_k]=0.$$ Also, $\mathbf{v}[k]\in\mathcal{S}_k$. Hence, MST applies, and we have
\begin{align}
&\sum_{k=1}^{T} K_{k-1}\widetilde{\nu}_R[k]\mathbf{v}^T[k]=\nonumber\\
&\begin{bmatrix}
    o(\sum_{k=0}^{T-1}v^2_1[k+1]) &  \dots  & o(\sum_{k=0}^{T-1}v^2_p[k+1]) \\
    o(\sum_{k=0}^{T-1}v^2_1[k+1]) &  \dots  & o(\sum_{k=0}^{T-1}v^2_p[k+1]) \\
    \vdots & \dots & \vdots \\
    o(\sum_{k=0}^{T-1}v^2_1[k+1]) &  \dots  & o(\sum_{k=0}^{T-1}v^2_p[k+1])
\end{bmatrix} 
+O(1),\label{po4}
\end{align}
where $v_i[k]$ denotes the $i^{th}$ component of $\mathbf{v}[k]$. Now, since $K_k\nu_R[k]=K_k\widehat{\nu}_R[k]+K_k\widetilde{\nu}_R[k]$, from (\ref{innov_mix_e}), we have
\begin{align}
K_{k-1}\nu_R[k]=(I-K_\nu)^{-1} Be[k-1]\nonumber\\
+(I-K_\nu)^{-1}K_{k-1}\widetilde{\nu}_R[k].
\end{align}
Substituting this in (\ref{po2}) and equating the diagonals using (\ref{po1}) and (\ref{po4}) completes the proof. 
\end{proof}

\section{Active Defense for Networked Cyber-Physical Systems: MIMO Systems with Gaussian Noise}\label{mimo}
In this section, we investigate multiple-input-multiple-output (MIMO) linear stochastic dynamical systems. They pose additional challenges as compared to SISO systems since malicious sensors in a MIMO system can attempt to collude so as to prevent the other nodes from detecting their presence. In this section, we show how the actuators can prevent the malicious nodes from introducing "excessive" distortion, lest they expose their presence. 

An $m$ input MIMO Linear Dynamical System of order $n$ is described by 
\begin{align}
\mathbf{x}[t+1]=A\mathbf{x}[t]+B\mathbf{u}[t]+\mathbf{w}[t+1]
\end{align}
where $\mathbf{x}[t]\in \mathbb{R}^n$, $\mathbf{u}[t]\in \mathbb{R}^m$, ${A}\in \mathbb{R}^{n\times n}$, ${B}\in \mathbb{R}^{n\times m}$, and $\{\mathbf{w}\}$ is a zero-mean i.i.d. sequence of Gaussian random vectors with covariance matrix $\sigma_w^2{I}_n$. 
We consider the case where $\mathbf{x}$ is perfectly observed. 

Let $\mathbf{z[t]}$ be the measurements reported by the sensors to the other nodes in the system. Note that since the nodes in the system are allowed to exchange reported measurements among them, all malicious sensors have to be consistent and report the same (but possibly erroneous) measurements to all honest nodes in the system. A history-dependent control policy is assumed to be in place, which dictates that the $i^{th}$ input to the system at time $t$ be
\begin{align}
u^g_i[t]=g^i_t(\mathbf{z}^t).
\end{align}

As reasoned before, to secure the system, each actuator superimposes on the input specified by the control policy, an additional zero-mean private excitation that it draws from a distribution, here Gaussian, that is made public. It should be noted that while the \emph{distribution} is made public, the actual values of the excitation are not revealed by the actuator to any other node. The private excitation value drawn at each time $t$ is chosen to be independent of its private excitation values at all the other time instants, of the private excitation values of other actuator nodes, and of the control policy-specified input. Therefore, actuator $i$ applies at time $t$ the input
\begin{align}
u_i[t]=u^g_i[t]+e_i[t]=g^i_t(\mathbf{z}^t)+e_i[t],\label{input_protocol}
\end{align}
where $e_i[t]\sim \mathcal{N}(0,\sigma_e^2)$ is independent of $e_j[k]$ for $(j,k)\neq (i,t)$, $\mathbf{x}[m], \mathbf{z}[m]$ for $m\leq t$, and $\mathbf{w}[n]$ for all $n$. 
We assume that all actuators are honest, meaning that they apply control inputs in accordance to (\ref{input_protocol}).  

We propose the following tests to be performed by each actuator $i$. 

\begin{itemize}
\item[1)] \textbf{Actuator Test 1:} Actuator $i$ checks if the reported sequence of measurements $\{\mathbf{z}\}$ satisfies 
\begin{align}
&\lim_{T\to\infty} \frac{1}{T}\sum_{k=0}^{T-1} (\mathbf{z}[k+1]-A\mathbf{z}[k]-Bg_k(\mathbf{z}^k))\nonumber\\
&(\mathbf{z}[k+1]-A\mathbf{z}[k]-Bg_k(\mathbf{z}^k))^T=\sigma_e^2BB^T+\sigma_w^2I_n\label{mimotest2}
\end{align}
\item[2)] \textbf{Actuator Test 2:} Actuator $i$ checks if the reported sequence of measurements $\{\mathbf{z}\}$ satisfies
\begin{align}
\lim_{T\to\infty} \frac{1}{T} \sum_{k=0}^{T-1} e_i[k] (\mathbf{z}[k+1]-A\mathbf{z}[k]-Bg_k(\mathbf{z}^{k})) \nonumber\\
= B_{\cdot,i}\sigma_e^2\label{mimotest3}
\end{align}
\end{itemize}
In Actuator Test 2, we have simplified the test to directly check the quantity that is important, which in this is case is the analog of (\ref{test2}). We note that for any set of conditions to be an admissible test, they have to be (i) checkable by the sensors and actuators based on their observations and what is reported to them, and (ii) satisfied if all parties are honest. The above two conditions satisfy these two conditions.

Define $$\mathbf{v}[t+1]\coloneqq \mathbf{z}[t+1]-A\mathbf{z}[t]-Bg_t(\mathbf{z}^t)-B\mathbf{e}[t]-\mathbf{w}[t+1],$$ so that if $\mathbf{z}[t]\equiv \mathbf{x}[t],$ $\mathbf{v}[t]\equiv 0.$ It is easy to see that as in the case of SISO systems, the sequence $\{\mathbf{v}\}$ has the intepretation as the additive distortion of the process noise deliberately introduced by the malicious sensors. We term the quantity $$\lim_{T\to\infty}\frac{1}{T}\sum_{k=1}^{T} ||\mathbf{v}[k]||^2$$ the \emph{additive distortion power} of the malicious sensors. The following theorem, akin to Theorem 1, proves that malicious sensors of only zero additive distortion power can pass the above tests, thereby remaining undetected.

\begin{theorem}
Suppose that the reported sequence of measurements passes the tests (\ref{mimotest2}), and (\ref{mimotest3}). If $B$ is of rank $n$, then, 
\begin{equation}
\lim_{T\to\infty}\frac{1}{T}\sum_{k=1}^{T} ||\mathbf{v}[k]||^2=0.
\end{equation}
\end{theorem}
\begin{proof}
Since $\{\mathbf{z}\}$ satisfies (\ref{mimotest3}), we have, $\forall i\in\{1,2,...,m\}$,
\begin{align}
\lim_{T\to\infty}\frac{1}{T}\sum_{k=1}^{T} {e}_i[k](B\mathbf{e}[k]+\mathbf{w}[k+1]+\mathbf{v}[k+1])\nonumber\\
=\sigma_e^2B_{\cdot(i)}.\nonumber
\end{align}
It follows that for all $i\in\{1,2,\cdots,m\}$,
\begin{align}
\lim_{T\to\infty}\frac{1}{T}\sum_{k=1}^{T} {e}_i[k]\mathbf{v}[k+1]=0.
\end{align}
Therefore,
\begin{align}
\lim_{T\to\infty}\frac{1}{T}\sum_{k=1}^{T} \mathbf{e}[k]\mathbf{v}^T[k+1]=0.\label{t2use1}
\end{align}
Since $\{\mathbf{z}\}$ also satisfies (\ref{mimotest2}), we have, 
\begin{align}
&\lim_{T\to\infty} \frac{1}{T}\sum_{k=0}^{T-1} (B\mathbf{e}[k]+\mathbf{w}[k+1]+\mathbf{v}[k+1])\nonumber\\
&(B\mathbf{e}[k]+\mathbf{w}[k+1]+\mathbf{v}[k+1])^T=\sigma_e^2BB^T+\sigma_w^2I_n
\end{align}
Using (\ref{t2use1}) and the fact that the process noise $\mathbf{w}[k+1]$, and the private excitation of the actuators at time $k$ are independent, the above simplifies to
\begin{align}
\lim_{T\to\infty} \frac{1}{T}\sum_{k=0}^{T-1} (\mathbf{w}[k+1]\mathbf{v}^T[k+1])+(\mathbf{w}[k+1]\mathbf{v}^T[k+1])^T\nonumber\\
(\mathbf{v}[k+1]\mathbf{v}^T[k+1])=0.\label{t2use2}
\end{align}

Define $S_{k}\coloneqq \sigma(\mathbf{x}^{k},\mathbf{z}^{k},\mathbf{e}^{k-2})$, and $\widehat{\mathbf{w}}[k]\coloneqq E[\mathbf{w}[k]\big |S_{k}]$. Since $$\mathbf{w}[k]=\mathbf{x}[k]-A\mathbf{x}[k-1]-Bg_{k-1}(\mathbf{z}^{k-1})-B\mathbf{e}[k-1]),$$ we have $$(\mathbf{x}^{k-2},\mathbf{e}^{k-2})\to(\mathbf{x}[k-1],\mathbf{x}[k],\mathbf{z}^{k})\to\mathbf{w}[k]$$ forming a Markov chain. Consequently, $\mathbf{\widehat{w}}[k]\coloneqq E[\mathbf{w}[k]\big |\sigma(\mathbf{x}^{k-2},\mathbf{e}^{k-2},\mathbf{x}[k-1],\mathbf{x}[k],\mathbf{z}^{k})]=E[\mathbf{w}[k]\big | \sigma(\mathbf{x}[k-1],\mathbf{x}[k],\mathbf{z}^{k})].$ Therefore,
\begin{align}
\widehat{\mathbf{w}}[k]=K_W(B\mathbf{e}[k-1]+\mathbf{w}[k]),
\end{align}
where $K_W\coloneqq \sigma_w^2(\sigma_e^2BB^T+\sigma_w^2I)^{-1}.$ Now define $\widetilde{\mathbf{w}}[k]\coloneqq \mathbf{w}[k]-\widehat{\mathbf{w}}[k]$. Then, $(\widetilde{\mathbf{w}}[k-1],\mathcal{S}_k)$ is a martingale difference sequence. This is because $\widetilde{\mathbf{w}}[k-1]\in\mathcal{S}_k$, and
\begin{align}
E[\widetilde{\mathbf{w}}[k]\big | S_{k}]=0.
\end{align}
Also, $\mathbf{v}[k]\in\mathcal{S}_k$. Hence, MST applies, and we have
\begin{align}
&\sum_{k=0}^{T-1} {\widetilde{\mathbf{w}}[k+1]\mathbf{v}}^T[k+1]=\nonumber\\
&\begin{bmatrix}
    o(\sum_{k=0}^{T-1}v^2_1[k+1]) &  \dots  & o(\sum_{k=0}^{T-1}v^2_p[k+1]) \\
    o(\sum_{k=0}^{T-1}v^2_1[k+1]) &  \dots  & o(\sum_{k=0}^{T-1}v^2_p[k+1]) \\
    \vdots & \dots & \vdots \\
    o(\sum_{k=0}^{T-1}v^2_1[k+1]) &  \dots  & o(\sum_{k=0}^{T-1}v^2_p[k+1])
\end{bmatrix} 
+O(1),\label{usef2}
\end{align}
where $v_i[k+1]$ denotes the $i^{th}$ element of $\mathbf{v}[k+1]$.

Using the above, we have
\begin{align}
\mathbf{w}[k+1]&=\widehat{\mathbf{w}}[k+1]+\widetilde{\mathbf{w}}[k+1]\nonumber\\
&=K_W(B\mathbf{e}[k]+\mathbf{w}[k+1])+\widetilde{\mathbf{w}}[k+1].
\end{align}
From the assumption on the rank of $B$, it follows that $K_W$ has all eigenvalues strictly lesser than unity. Therefore, simplifying the above, 
\begin{align}
\mathbf{w}[k+1]=(I-K_W)^{-1}K_WB\mathbf{e}[k]\nonumber\\
+(I-K_W)^{-1}\widetilde{\mathbf{w}}[k+1].
\end{align}
Substituting this into (\ref{t2use2}), using (\ref{t2use1}) and (\ref{usef2}), and equating the $q^{th}$ entry along the diagonal, we have
\begin{align}
\sum_{k=0}^{T-1} v^2_q[k+1] + o(\sum_{k=0}^{T-1} v^2_q[k+1])=o(T).
\end{align}
Since this is true for all $q\in\{1,2,...,p\}$, dividing the above by $T$ and taking the limit as $T\to\infty$ completes the proof.
\end{proof}

\section{Extension to Non-Gaussian Systems}\label{nongaussian}
The results developed in the previous sections assumed that the process noise follows a Gaussian distribution. This assumption can be relaxed to a certain extent. In this section, we illustrate how this may be relaxed for a single-input-single-output system. 

Consider a SISO system described by 
\begin{align}
x[t+1]=ax[t]+bu[t]+w[t],
\end{align}
where $w[t]\sim P_W$ is an i.i.d. process with mean 0 and variance $\sigma_w^2$. In such a case, the actuator can choose the distribution of its private excitation such that the output of the private excitation has the same distribution as the process noise. Hence, the actuator applies to the system the input
\begin{align}
u[t]=g_t(z^t)+e[t],
\end{align}
where $be[t]\sim P_W$ is an i.i.d. sequence. As before, even though the actuator implements the control policy $\{g\}$, the policy is applied to the measurements $z[t]$ reported by the sensor, which could differ from the true output $x[t]$. Therefore, the system evolves in closed-loop as
\begin{align}
x[t+1]=ax[t]+bg_t(z^t)+be[t]+w[t].
\end{align}

The actuator then performs the following tests to check if the sensor is malicious or not.   
\begin{itemize}
\item[1)] \textbf{Actuator Test 1:} Check if the reported sequence of measurements $\{z[t]\}$ satisfies
\begin{flalign}
&\lim_{T\to\infty} \frac{1}{T} \sum_{k=0}^{T-1}(z[k+1]-az[k]-bg_k(z^k)-be[k])^2\nonumber\\
&=\sigma_w^2.\label{test2_nong}
\end{flalign}

\item[2)] \textbf{Actuator Test 2:} Check if the reported sequence of measurements $\{z[t]\}$ satisfies
\begin{flalign}
&\lim_{T\to\infty} \frac{1}{T} \sum_{k=0}^{T-1} (z[k+1]-az[k]-bg_k(z^k))^2\nonumber\\ 
&=2\sigma_w^2.\label{test1_nong}
\end{flalign}
\end{itemize}
As before, let $$v[t+1]\coloneqq z[t+1]-az[t]-bg_t(z^t)-be[t]-w[t+1],$$ so that for an honest sensor which reports $z[t]\equiv x[t]$, $v[t]=0\;\; \forall t$. The \emph{additive distortion power} of a malicious sensor too is defined as before, i.e., as
$$\lim_{T\to\infty} \frac{1}{T} \sum_{k=1}^{T} v^2[k].$$
The following result, a generalization of Theorem 1, shows that the above tests suffice to ensure that a malicious sensor of only zero effective power can remain undetected.  

\begin{theorem}
If $\{z[t]\}$ satisfies tests (\ref{test2_nong}) and (\ref{test1_nong}), then, 
\begin{equation}
\lim_{T\to\infty} \frac{1}{T} \sum_{k=1}^{T} v^2[k]=0.\label{thm1a_nong} 
\end{equation}
\end{theorem}

\begin{proof}
Following the same sequence of arguments as in the proof of Theorem 1, we arrive at the following equalities: 

\begin{align}
\lim_{T\to\infty} \frac{1}{T} \sum_{k=1}^{T} v^2[k]+ \lim_{T\to\infty} \frac{1}{T} \sum_{k=1}^{T} 2v[k]w[k]=0, \label{t1use2_g}
\end{align}

\begin{align}
\lim_{T\to\infty} \frac{1}{T} \sum_{k=1}^{T} e[k-1]v[k]=0. \label{key1_g}
\end{align}

Let $\mathcal{S}_k\coloneqq \sigma (x^k,z^k,e^{k-2})$, and $\widehat{w}[k]\coloneqq E[w[k]\big | \mathcal{S}_k]$. The general complication is that due to non-Gaussianity, the conditional mean estimate may be non-linear in $e[k-1]$ and $w[k]$, unlike in the Gaussian case. However, since the sequence of observations are i.i.d., and the distribution of the actuator's private excitation is chosen to be the same as the distribution of the process noise, we have
\begin{equation}
\widehat{w}[k]=\frac{1}{2}(be[k-1]+w[k])=\beta (be[k-1]+w[k]),\label{estimate_form}
\end{equation}
where $\beta\coloneqq \frac{1}{2}$. This can be written as 
\begin{equation*}
\widehat{w}[k]=\alpha e[k-1] + \beta w[k], 
\end{equation*}
where $\alpha\coloneqq b\beta$. Let $\widetilde{w}[k]\coloneqq w[k]-\widehat{w}[k]$. 

Now, the following result also holds following the same sequence of arguments employed in the proof of Theorem 1:
\begin{equation}
\sum_{k=1}^{T} v[k]\widetilde{w}[k]=o(\sum_{k=1}^{T} v^2[k]) + O(1).
\end{equation}

Hence, following similar arguments as in the proof of Theorem 1, we obtain
\begin{align}
&\sum_{k=1}^{T} v^2[k]+ \sum_{k=1}^{T} 2v[k]w[k]=(1+o(1))(\sum_{k=1}^{T} v^2[k]) + O(1).\nonumber
\end{align}
Dividing the above equation by $T$, taking the limit as $T\to\infty$, and invoking (\ref{t1use2_g}) completes the proof. 
\end{proof}

\noindent\textbf{Remark:} The net import of this result is that in order to safeguard against malicious actions of the sensor we need to add a level of noise exactly equal to the process noise already present in the system. Hence there is an amplification of the process noise's standard deviation by $\sqrt{2}$ that appears the price of guarding against malicious attack by the sensor. In contrast, in the Gaussian case, we have seen that this extra cost can be made as small as desired by choosing $\sigma_e^2$ as small as desired, though of course at the cost of delaying detection with an acceptable false alarm probability, as we discuss in the next section.

\section{Statistical Tests for Active Defense}\label{statTests}
The results developed in the previous sections are couched in an asymptotic fashion. They characterize what can be detected and how that may be done. These asymptotic characterizations can be used to develop statistical tests that reveal malicious activity in a finite period of time with acceptable false alarm rates. 

The task at hand is to develop statistical tests equivalent to asymptotic tests such as (\ref{test2}), (\ref{test1}), (\ref{mimotest2}), (\ref{mimotest3}), which include tests for covariance matrices of certain random vectors. This problem has been addressed in \cite{mehra_71} in the context of fault detection in control systems, in which the innovation sequence is tested for the properties of whiteness, mean, and covariance. Since these are also the properties that have to be satisfied by the test sequences that we construct in (\ref{test2}), (\ref{test1}), and (\ref{mimotest3}), the same test described in \cite{mehra_71} can be used. 



Yet another test that is particularly well-suited for the problem at hand is the sequential probability ratio test, introduced in \cite{wald_sequential}. In the standard setting of sequential hypothesis testing, for each time $t$, the set of all possible observations $S_m$ (or the space of sufficient statistics) is partitioned into three sets $R_0$, $R_1$, $R$. At each time $t$, a decision is made whether to select the null hypothesis (and stop), select the the alternative hypothesis (and stop), or continue observing. The null (alternative) hypothesis is accepted and hypothesis testing is stopped if at some time $t$, the observations available till time $t$, denoted by $\{z_s^t\}$, fall in the set $R_0$ ($R_1$). Hypothesis testing is continued at time $t$ if $\{z_s^t\}$ falls in $R$. One such sequential test is the sequential probability ratio test presented in \cite{wald_sequential}, where the likelihood ratio of observations up to time $t$, denoted by $\Lambda_t$, is compared against two thresholds $\tau_1$ and $\tau_2$. The null (alternative) hypothesis is accepted at time $t$ if $\Lambda_t<\tau_1$ ($\Lambda_t>\tau_2$). If $\tau_1\leq\Lambda_t\leq\tau_2$, another observation is drawn, and the process repeats.

There are a few aspects where the problem at hand departs from the above. These are enumerated below. In what follows, the null hypothesis is that the control system is not under attack.
\begin{itemize}
\item[1)] In the problem of secure control, it is not possible to attribute a distribution (or even a parameterized, finite-dimensional family of distributions) from which the observations would occur under the alternative hypothesis (that the control system is under attack). Consequently, a likelihood ratio cannot be defined, and the sequential probability ratio test cannot be used. Hence, the problem at hand is not to test one hypothesis against another, but is only that of rejecting the null hypothesis or otherwise.
\item[2)] In the classical setup as described in \cite{wald_sequential}, in a finite time and with probability 1, the test stops, and one hypothesis or the other is accepted. However, the adversary could be a "sleeper agent". It could behave as an honest party for a long period of time, and then operate in a malicious manner. Hence in our situation, at no time can one "accept" the null hypothesis forever for all subsequent time. Therefore, in the systems of interest here, there are only two possible decisions at any time $t$, reject the null hypothesis or continue observations. 
\item[3)] In the classical setting, optimal performance is obtained by considering all observations that are available up to the time of making a decision. However, in our problem, an adversary that behaves maliciously in a bursty fashion, is not likely to be detected since long term averaging would nullify these effects. In order to account for this, a moving window approach could be employed over which a test statistic can be calculated. That is, the statistical test is conducted over a window of $l$ observations, $nl \leq t < (n+1)l$. In each window it assesses whether to reject the null hypotheses. In this manner, malicious activity can be detected within $l$ samples, with acceptable false alarm rate. The choice of the excitation variance $\sigma_e^2$ can be based on the acceptable detection delay $l$, and the false alarm rate. 
\end{itemize}
Based on the above, the basis of a statistical hypothesis test for the problem can be chosen as follows. Under the null hypothesis, the sample covariance matrix follows the Wishart distribution \cite{anderson_book} (the multidimensional counterpart of the Chi-squared distribution). Hence, for a given false alarm rate $\alpha$, the threshold $\tau(\alpha)$ can be determined for the likelihood function of the observations, where the likelihood function considered for the test at time $t$ is the probability density evaluated at the $l$ most recent observations. If this likelihood function exceeds the threshold at any time $t$, an alarm is raised. If not, the test is repeated for the next time instant. The following simulation example is illustrative. 

\textbf{Example:} We consider the scenario of an ARX system addressed in Section-\ref{sisoarx}. Specifically, we consider a second order plant of the type common in process control. A single-input, single-output stable plant obeying (\ref{siso_arx_model}), with parameters $a_0=0.7$, $a_1=0.2$, $b_0=1$, $b_1=0.5$ and $\sigma_w^2=1$ is considered. Note that the chosen system is minimum phase, so that the pre-equalizer required to filter the sequence of private excitation is stable. 

The sensor is assumed to be initially well-behaved, but suddenly become malicious after a certain period (denoting the time of initiation of attack), while the actuator is assumed to remain uncompromised. The actuator applies control inputs $$u[k]=-\frac{1}{b_0}(a_0z[t]+a_1z[t-1]+b_1u[k-1])+e[k],$$ where $z[k]$ is the measurement reported by the sensor at time $k$, and $e[k]\sim$ i.i.d. $\mathcal{N}(0,1)$ is the actuator node's private excitation. Therefore, the system evolves as 
\begin{align*}
y[k+1]=0.7(y[k]-z[k])+0.3(y[k-1]-z[k-1])\\
+e[k]+w[k+1].
\end{align*}

The adversary attacks the system at some random time and begins to report false measurements. In order to do so, the adversary estimates optimally the process noise at each instant from the measurements that it observes. Since the process noise and the private excitation have the same variance, the optimal estimate is simply $\frac{1}{2}(y[k+1]-0.7(y[k]-z[k])-0.3(y[k-1]-z[k-1]).$ Once the adversary forms its estimate, it adds a noise correlated with the estimate as follows. The adversary forms $v[t]=n[t]-\widehat{w}[t]$, where $n[t]\sim \mathcal{N}(0,1)$, the same distribution as $w[t]$, and reports $z[t]=x[t]+v[t]$. Note that in the absence of dynamic watermarking, employing this strategy would enable the adversary to form perfect estimates of the process noise, and consequently, as reasoned in Section-\ref{solution}, would render every detection algorithm useless. 

Fig. \ref{stat} plots the evolution of the negative log likelihood function over time when dynamic watermarking is employed. In our simulation, the adversary initiates attack at time epoch $4500$. As indicated in Fig. \ref{stat}, the (windowed) negative log likelihood function corresponding to (\ref{test2_arx}) stays within limits until the attack. Around time epoch 4500, the likelihood function steadily increases, indicating the onset of an attack. Based on tolerable false alarm rates, an appropriate detection threshold can be set. 

\begin{figure}
\centering
\includegraphics[width=\columnwidth]{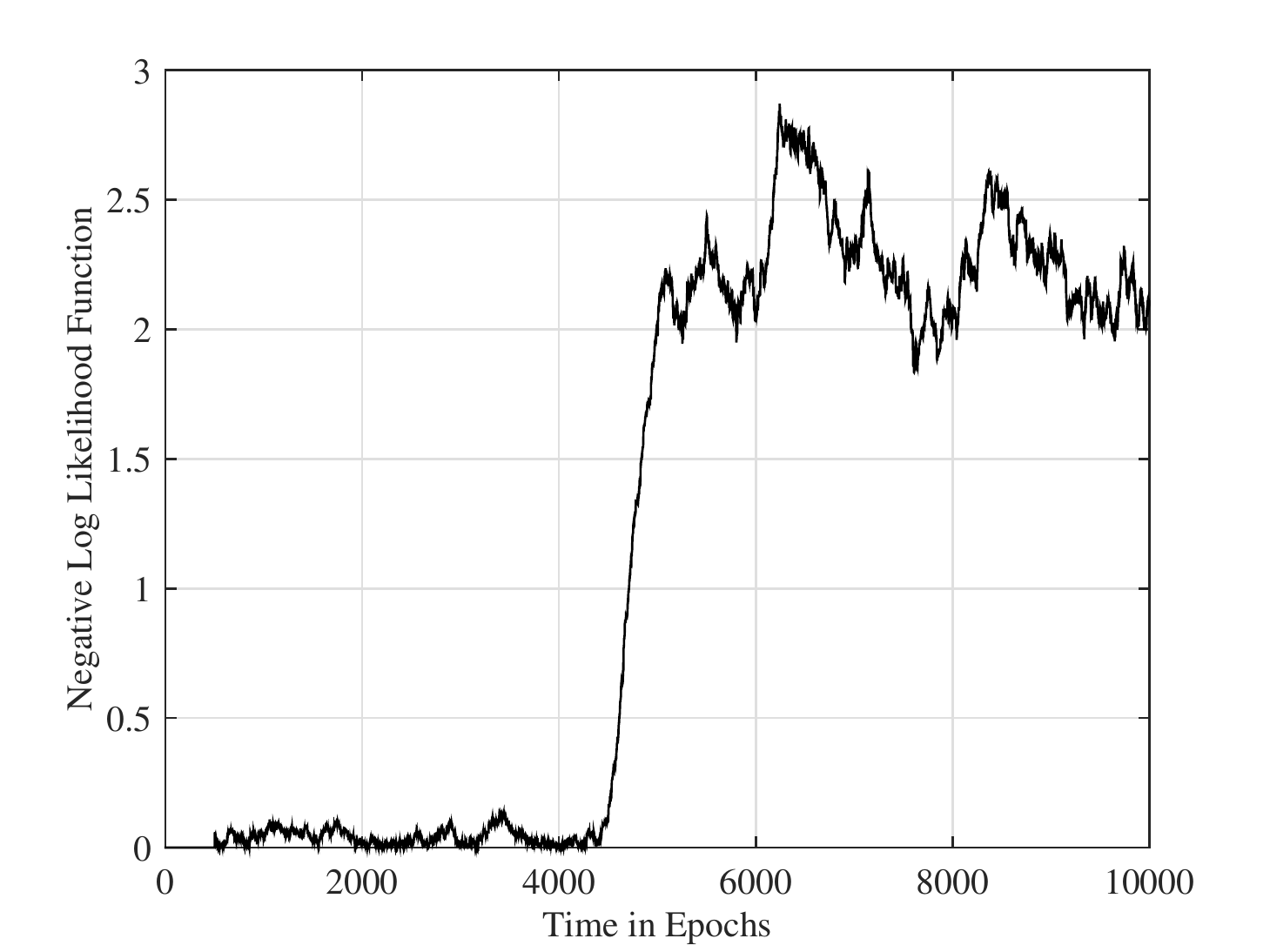}
\caption{Sequential Hypothesis Testing for Attack Detection}\label{stat}
\end{figure}

\section{Conclusion}\label{conclusion}
In this paper, we have considered the fundamental problem of security of networked cyber-physical systems. Securing such systems is not the same as securing a communication network due to the fundamental role played by sensors and actuators in interfacing with and controlling the physical plant. We have provided a general procedure for dynamic watermarking that imposes private excitation signals on the actuation signals whose presence can then be traced around the loop to detect malicious behaviors of sensors. 

We have only explored the tip of the iceberg of this procedure. It can be used in more general contexts for a variety of plants of interest in control systems such as  process control systems where more nonlinear models may be of interest. Further study into performance of finite-time tests for Dynamic Watermarking would also be of practical interest.

There is much to be done in this area if cyber-physical system applications are to proliferate to address critical infrastructural and societal needs.

\section*{Addendum}
\noindent \textbf{Notation:} This paper uses the following notation. 
\begin{itemize}
\item Scalars are denoted using lowercase: $a$
\item Vectors are denoted using lowercase boldface: $\mathbf{x}$
\item The $i^{th}$ component of vector $\mathbf{x}$: $x_i$
\item Matrices are denoted using uppercase: $A$
\item The $i^{th}$ column of matrix $A$: $A_{\cdot i}$
\item The $i^{th}$ row of matrix $A$: $A_{i\cdot}$
\item The submatrix of $A$ formed by the intersection of rows $i$ through $j$ and columns $p$ through $q$: $A_{i:j,p:q}$
\item The matrix with $i^{th}$ column of $A$ removed: $A_{\cdot (-i)}$
\item The matrix with $i^{th}$ row of $A$ removed: $A_{(-i)\cdot}$
\item The vector with $i^{th}$ component of $\mathbf{x}$ removed: $\mathbf{x}_{-i}$
\end{itemize}

\section*{acknowledgments}
The authors would like to thank Bruno Sinopoli (CMU), Yilin Mo (NTU) and Sean Weerakkody (CMU) for their valuable comments.


%

\ifCLASSOPTIONcaptionsoff
  \newpage
\fi



%
\bibliographystyle{IEEEtran}
\bibliography{references}

\begin{IEEEbiography}
[{\includegraphics[width=1in,height=1.25in,clip,keepaspectratio]{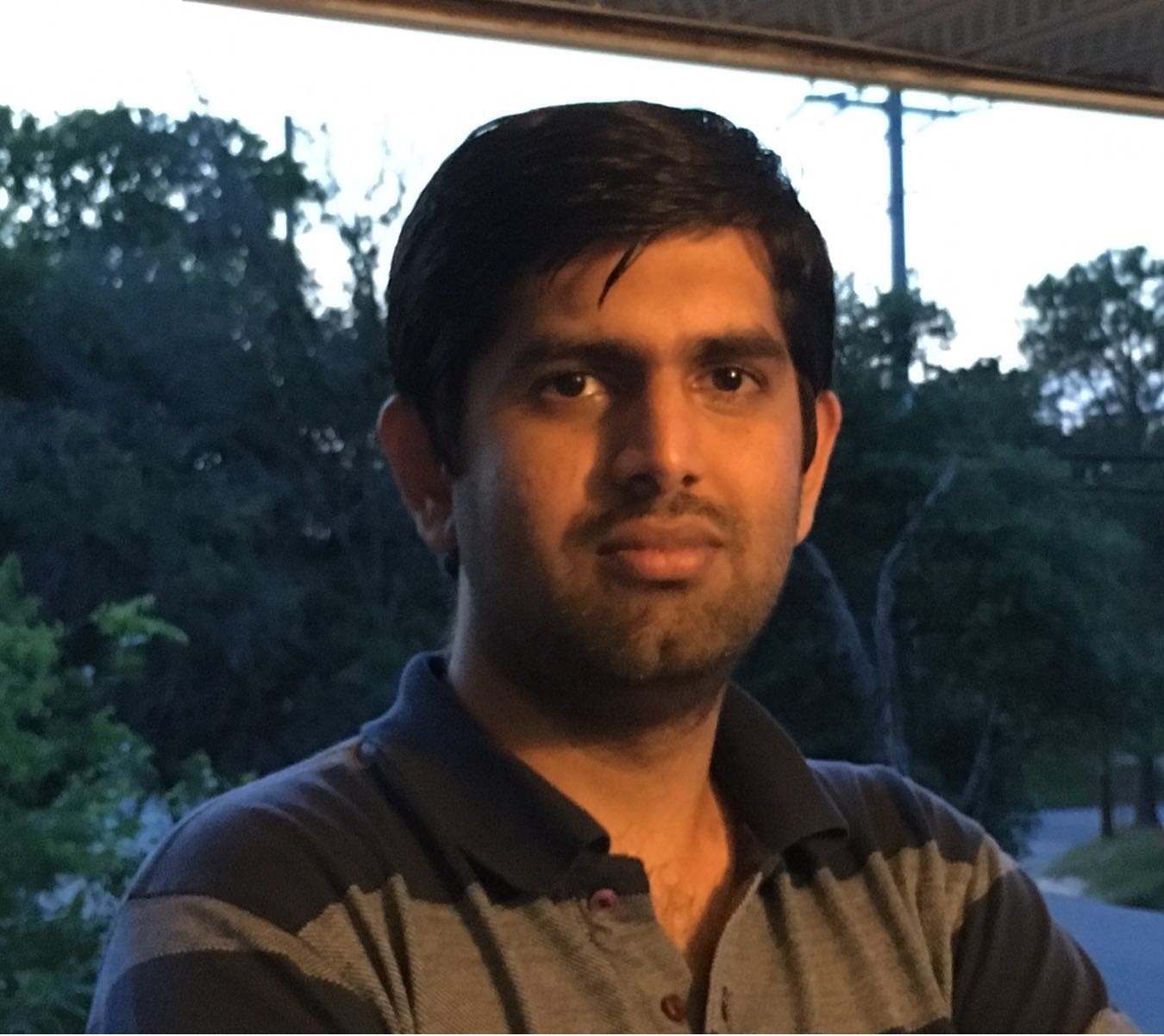}}]{Bharadwaj Satchidanandan} is a doctoral student in the Electrical and Computer Engineering department at Texas A\&M University, College Station, TX. 
Prior to this, he completed his Master's from Indian Institute of Technology Madras, where he worked on Wireless Communications. 
Between May '15 and August '15, he interned at Intel Labs, Santa Clara, CA, where he worked on interference cancellation algorithms for next-generation wireless networks.
His research interests include cyberphysical systems, power systems, security, database privacy, communications, control, and signal processing.
\end{IEEEbiography}

\begin{IEEEbiography}
[{\includegraphics[width=1in,height=1.25in,clip,keepaspectratio]{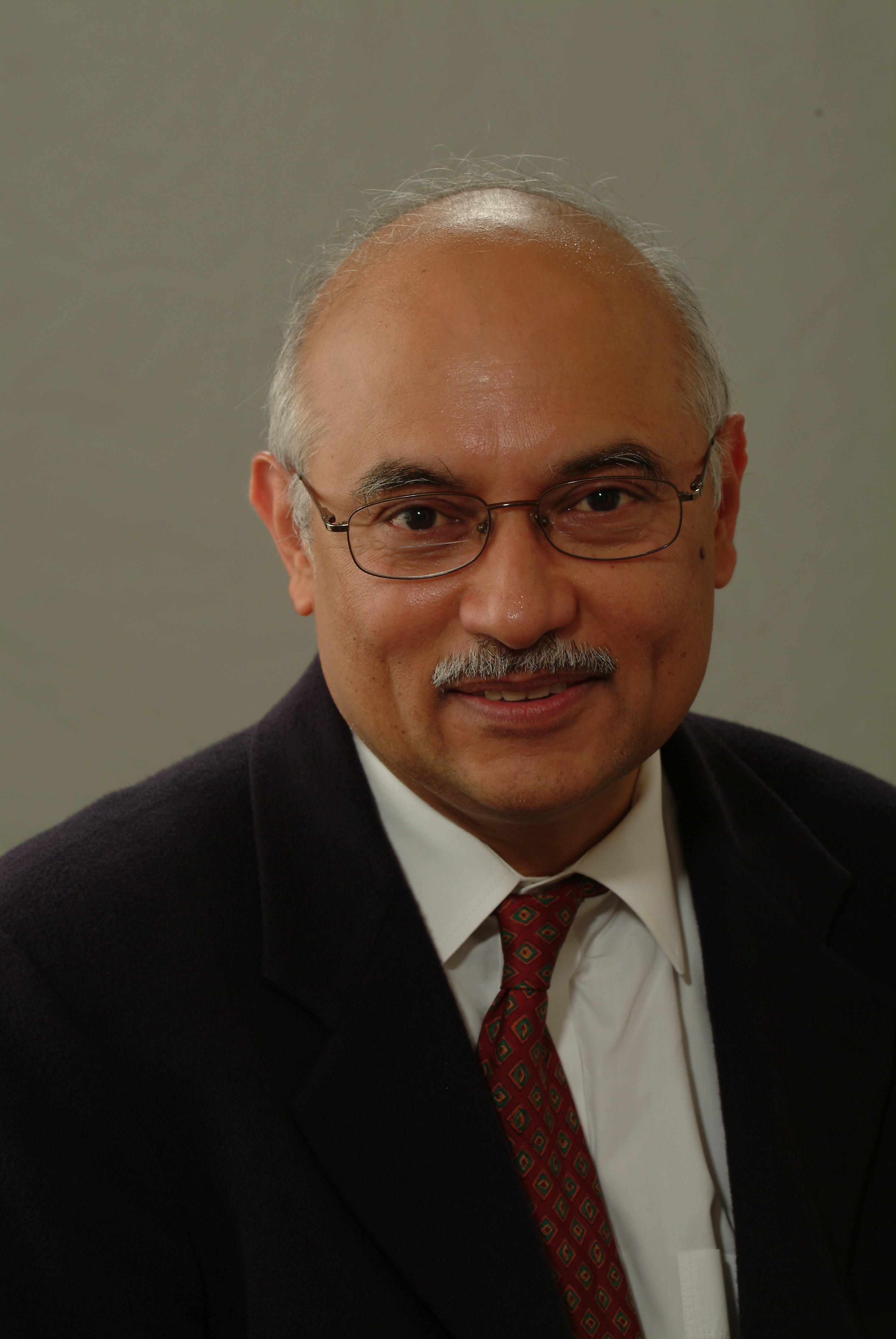}}]{P. R. Kumar} B.  Tech.  (IIT  Madras,  ‘73),  D.Sc.(Washington  University,  St.  Louis,  ‘77),  was  a faculty member at UMBC (1977-84) and Univ. of Illinois,  Urbana-Champaign  (1985-2011).  He  is currently  at  Texas  A\&M  University.  His  current research is focused on stochastic systems, energy systems, wireless networks, security, automated transportation, and cyberphysical systems. He is a member of the US National Academy of Engineering and The World Academy of Sciences. He was awarded a Doctor Honoris Causa by  ETH,  Zurich.  He  has  received  the  IEEE  Field  Award  for  Control Systems, the Donald P. Eckman Award of the AACC, Fred W. Ellersick Prize of the IEEE Communications Society, the Outstanding Contribution Award of ACM SIGMOBILE, the Infocom Achievement Award, and the SIGMOBILE  Test-of-Time  Paper  Award.  He  is  a  Fellow  of  IEEE  and ACM  Fellow.  He  was  Leader  of  the  Guest  Chair  Professor  Group  on Wireless  Communication  and  Networking  at  Tsinghua  University,  is  a D.  J.  Gandhi  Distinguished  Visiting  Professor  at  IIT  Bombay,  and  an Honorary Professor at IIT Hyderabad. He was awarded the Distinguished Alumnus Award from IIT Madras, the Alumni Achievement Award from Washington Univ., and the Daniel Drucker Eminent Faculty Award from the College of Engineering at the Univ. of Illinois.
\end{IEEEbiography}

%

\end{document}